\documentclass[11pt,a4paper]{article}
\usepackage{t1enc}
\usepackage[latin1]{inputenc}
\usepackage[english]{babel}
\normalfont
\usepackage{amsmath}
\usepackage{amssymb}
\usepackage{amsthm}
\usepackage{yfonts}
\let\ltxcup\cup
\usepackage{mathabx}
\usepackage{bbm}
\usepackage{bm}
\usepackage{wasysym}

\usepackage{mathrsfs}
\usepackage{pifont}
\usepackage{hyperref}

\newcommand{\be}[0]{\begin{equation}}
\newcommand{\ee}[0]{\end{equation}}

\numberwithin{equation}{section}

\theoremstyle{plain}% default
\newtheorem{theorem}{Theorem}[section]
\newtheorem{lemma}[theorem]{Lemma}
\newtheorem{proposition}[theorem]{Proposition}
\newtheorem{corollary}[theorem]{Corollary}

\newtheorem{assumption}[theorem]{Assumption}

\theoremstyle{definition}

\begin{document}

\begin{center}
{\Large 
{\bf The global anomaly of the self-dual field in general backgrounds}}
\vspace{2.0cm}

{\large Samuel Monnier}
\vspace*{0.5cm}

Institut für Mathematik, 
Universität Zürich,\\
Winterthurerstrasse 190, 8057 Zürich, Switzerland\\
samuel.monnier@gmail.com 

\vspace*{2cm}

{\bf Abstract}
\end{center}

We prove a formula for the global gravitational anomaly of the self-dual field theory in the presence of background gauge fields, assuming the results of \href{http://arxiv.org/abs/1110.4639}{arXiv:1110.4639}. Along the way, we also clarify various points about the self-dual field theory. In particular, we give a general definition of the theta characteristic entering its partition function and settle the issue of its possible metric dependence. We treat the cohomological version of type IIB supergravity as an example of the formalism: a mixed gauge-gravitational global anomaly, occurring when the B-field and Ramond-Ramond 2-form gauge fields have non-trivial Wilson lines, cancels provided a certain cobordism group vanishes.

\newpage

\tableofcontents

\section{Introduction}

Global gravitational anomalies \cite{Witten:1985xe, Monnier:2012pn} detect anomalous phases that the partition function of a chiral quantum field theory can pick under large diffeomorphisms. For a quantum field theory supposed to describe the low energy limit of a quantum theory of gravity, such anomalous phases signal a fundamental inconsistency. Global gravitational anomaly cancellation, as a generalization of the well known local anomaly cancellation, provides non-trivial constraints on quantum field theories appearing in various effective descriptions of string theory or M-theory. For the conformal field theories on the string worldsheet, it takes the form of the familiar modular invariance constraint.

While the global gravitational anomalies of chiral fermionic field theories were understood a long time ago \cite{Witten:1985xe}, a formula for the global anomaly of the self-dual field theory on generic manifolds was derived only recently \cite{Monnier2011, Monnier2011a}. %\footnote{The formula for the global anomaly of the self-dual field theory proposed in \cite{Witten:1985xe} is valid only when the manifold has no middle-degree cohomology, and does not make anomaly cancellation manifest in type IIB supergravity \cite{Monnier2011a}.} 
A self-dual $2\ell$-form field on a $4\ell+2$-dimensional manifold can be naturally sourced by a $2\ell+1$-form abelian gauge field. The derivation of \cite{Monnier2011, Monnier2011a} assumes that this background gauge field vanishes. While consistency requires its field strength to vanish, it can have non-trivial Wilson lines and torsion fluxes. The aim of this paper is to investigate how these non-trivial backgrounds affect the global anomaly and to find a general formula for the latter. The formula studied here was already sketched roughly in Section 4.4 of \cite{Monnier2011a}.

The derivation in \cite{Monnier2011, Monnier2011a} was not completely rigorous. Several consistency checks of the proposed anomaly formula were presented, one of them being the cancellation of global gravitational anomalies in the cohomological version of type IIB supergravity (see Section \ref{SecIIB}). In this paper, we propose a general formula and we show rigorously that under the assumption that the formula conjectured in \cite{Monnier2011a} is correct, then the general formula holds. The assumption and its validity are discussed in detail in Section \ref{SecHorizLoopVanGaugField}. The derivation of a formula valid in any background is a necessary step toward systematic checks of global gravitational anomaly cancellation in string theory, M-theory and supergravity. In \cite{Monnierb}, we show that the formula derived here implies the cancellation of global gauge and gravitational anomalies for the worldvolume theory of chiral five-branes, when the anomaly inflow from the bulk is taken into account.\\

Apart from the anomaly formula, the main outcome of our investigation is the clarification of several puzzles in the literature on the self-dual field theory. The partition function of the self-dual field involves a theta function whose characteristic $\tilde{\eta}$ has to be specified. Let $M$ be the manifold of dimension $4\ell+2$ on which the self-dual field is defined. The relevant theta function is essentially a Siegel theta function defined on the Kähler space $H^{2\ell+1}(M,\mathbb{R})$, with the symplectic and complex structures given by the wedge product pairing and the Hodge star operator. $H^{2\ell+1}(M,\mathbb{R})$ should be seen here as the space of Wilson lines of the background abelian $2\ell+1$-form gauge field. The theta characteristic is an element of $H^{2\ell+1}(M,\mathbb{R})$, modulo the lattice $\Lambda$ of classes with integral periods. It was shown in \cite{Monnier2011} that $\tilde{\eta}$ could essentially be identified with the torsion component of the integral first Chern class of the anomaly line bundle of the self-dual field theory, implying that $\tilde{\eta}$ can take only discrete values: $\tilde{\eta} \in \frac{1}{2}\Lambda$.

In \cite{Witten:1996hc}, the theta characteristic was determined by integrating a characteristic form on a manifold $W$ bounded by $M \times S^1$. As the integral takes place on a manifold with boundary, such a prescription does not ensure that $\tilde{\eta} \in \frac{1}{2}\Lambda$, or even that it is independent on the boundary data. It was indeed argued in \cite{Moore:2004jv,Moore2004} that the characteristic should depend continuously on the background metric.

In \cite{Witten:1999vg} a different prescription was used to determine $\tilde{\eta}$, this time using the integral of a characteristic class on $W$. While this formulation obviously produces an element of $\frac{1}{2}\Lambda$, its compatibility with the previous prescription was not clear for the reasons mentioned above. Other works like \cite{Belov:2006jd,Monnier2011} took an agnostic viewpoint and left this parameter unspecified.

The theta characteristic $\tilde{\eta}$ governs the global anomaly of the self-dual field theory with respect to the large gauge transformations of the background gauge field. Therefore our general global anomaly formula, when appropriately specialized, provides a formula for $\tilde{\eta}$, which turns out to be closely to \cite{Witten:1996hc}: it coincides with the adiabatic limit of the formula of \cite{Witten:1996hc}. We show that the puzzles mentioned above disappear: $\tilde{\eta}$ indeed belongs to $\frac{1}{2}\Lambda$ and it is metric independent. These properties follow essentially from the fact that the background gauge field, as a source of the self-dual field, has to be exact, and is trivialized by the latter. We also show how to recover a torsion anomaly analogous to the Freed-Witten anomaly, originally discovered in the formalism of \cite{Witten:1999vg}, and how it can be canceled by torsion fluxes.\\

Another outcome is the clarification of technical details concerning the computation of mixed gauge-gravitational global anomalies. Mixed gauge-gravitational anomalies are associated to loops in the quotient $\mathcal{F} = \mathcal{B}/\mathcal{G}$, where $\mathcal{B}$ is the space of background metrics and gauge fields, and $\mathcal{G}$ is the group of diffeomorphisms and gauge transformations. A loop $c \subset \mathcal{F}$ determines an equivalence class of metric and gauge field on each fiber of a mapping torus $M \rightarrow M_c \rightarrow S^1$. The global anomaly is essentially given by the evaluation on $M_c$ of a Chern-Simons functional depending on a metric and a gauge field on $M_c$. One therefore has to extend both the metric and the gauge field from the fibers of $M_c$ to $M_c$ itself. Doing this requires picking a specific gauge, or equivalently lifting the loop $c \subset \mathcal{F}$ to a path $p \subset \mathcal{B}$ (Lemma \ref{LemConstrDiffCocMapTor}). A priori, it is far from clear that the Chern-Simons functional is independent of this choice of lift. Proposition \ref{PropHolIndepPathLift} shows that this is in fact the case. We focus here on the Chern-Simons functional associated with the self-dual field, but this type of arguments should generalize to other chiral theories.
\\

As an example, we apply our formalism to the cohomological version of type IIB supergravity, namely the one in which the fluxes of the Ramond-Ramond gauge fields are classified by integral cohomology. This is a chiral theory in dimension 10 containing a self-dual field, whose global anomalies have not been completely studied. We find that while the pure global gravitational anomaly cancel, as was already noted in \cite{Monnier2011a}, an apparent mixed gauge-gravitational global anomaly is present. This is in fact not surprising, as it is a direct extension of the pure global gauge anomaly which was implicitly present in the original treatment of \cite{Witten:1996hc}. This anomaly vanishes if and only if a certain cobordism group vanishes. Infortunately, the computation of this cobordism group does not seem easy. %We offer a picture of how this anomaly should be canceled, but a missing mathematical result prevents us from checking the cancellation explicitly. 
\\

The paper is organized as follows. In Section \ref{SecAnomBundl}, we describe the anomaly bundle of the self-dual field theory \cite{Monnier2011} over the space of metrics and gauge fields modulo diffeomorphisms and gauge transformations. We point out certain crucial properties satisfied by the holonomies of the natural connection defined on the anomaly bundle. In Section \ref{SecHSForm}, we describe a formula for the global anomaly of the self-dual field in an arbitrary background. In Section \ref{SecProofAnForm}, assuming the results of \cite{Monnier2011a}, we show that it correctly reproduces the holonomies of the anomaly bundle. The example of type IIB supergravity is treated in Section \ref{SecIIB}. The beginning of each section contains a more detailed summary of their respective contents.

\section{The anomaly bundle of the self-dual field}

\label{SecAnomBundl}

We start by reviewing the geometrical picture of anomalies in Section \ref{SecGeomIntAnom}. We present the model of abelian gauge fields by (shifted) differential cocycles in Section \ref{SecAbGFAsDC}, which allows us to construct the space of background fields $\mathcal{B}$ of the self-dual field theory in Section \ref{SecSpaBackFields}. The background fields consist of the Riemannian metric of the underlying manifold $M$ and an abelian gauge field sourcing the self-dual field. The group $\mathcal{G}$ of diffeomorphisms and gauge transformations acts on $\mathcal{B}$, as explained in Section \ref{SecActGrpLocTrans}. Section \ref{SecRelModGeom} shows how modular geometry appears in relation to the intermediate Jacobian of $M$. Section \ref{SecAnomBundle} describes the anomaly line bundle of the theory as a pull-back of a theta line bundle. Its connection is constructed in Section \ref{SecConnAnBun}. In Section \ref{SecFactPathHol}, we point out certain properties of the fundamental group of the quotient $\mathcal{B}/\mathcal{G}$ which will ease the proof of the anomaly formula in Section \ref{SecProofAnForm}.

\subsection{The geometric interpretation of anomalies}

\label{SecGeomIntAnom}
 
Let us quickly remind the reader about anomalies and their geometrical interpretation. Anomalies are the breaking of a (global) symmetry of a classical field theory upon quantization. 
In the classical theory, we can associate to the symmetry a Noether current $J$, whose divergence vanishes upon imposing the equations of motion. In the quantum theory, the breaking of the global symmetry is characterized by a non-zero expectation value of the correlation functions involving the divergence of $J$.

In order to compute these correlation functions, and more generally to understand the realization of the global symmetry in the quantum field theory, we can use the background field method. We therefore introduce a background field acting as a source for $J$. %While the global symmetries leave the background field invariant, the corresponding local symmetries act on the space of background field. The expectation value of the divergence of the current is then given by the variation of the partition function with respect to an infinitesimal local symmetry. 
For gravitational anomalies, associated with a breaking of the diffeomorphism symmetry, the current is the energy-momentum tensor and the background field is the metric of the underlying manifold. For gauge anomalies, the current is the gauge current and the background field is a background gauge field. 

The partition function $Z$ can be seen as a function over the space $\mathcal{B}$ of background fields. The group of local transformations $\mathcal{G}$ associated to the global symmetry group acts on $\mathcal{B}$, and the presence of an anomaly translates in the lack of invariance of $Z$ under the action of $\mathcal{G}$. For instance, the expectation value of the divergence of the current can be obtained by computing the variation of the partition function under the action of an infinitesimal local transformation on $\mathcal{B}$. More generally, a lack of invariance under local transformations connected to the identity is associated to a local anomaly, while a lack of invariance under transformations in the other components of $\mathcal{G}$ signals a global anomaly. If the symmetry under consideration is to be gauged, then consistency clearly requires the partition function to be invariant under the whole group of local transformations; in other words the local and global anomalies have to vanish.

If $Z$ is not invariant under $\mathcal{G}$, it does not push down to a function over the quotient $\mathcal{B}/\mathcal{G}$. However it pushes down in all cases to a section of a line bundle $\mathscr{A}$ over $\mathcal{B}/\mathcal{G}$, the anomaly line bundle (see for instance \cite{Monnier:2012pn} for more details). In all cases at hand, $\mathscr{A}$ comes equipped with a natural connection $\nabla_\mathscr{A}$, and it turns out to be easier to study this connection than the (lack of) invariance of $Z$. The global anomaly, of interest to us in this paper, is the set of holonomies of $\nabla_\mathscr{A}$; it is precisely the obstruction to the geometric trivialization of $\mathscr{A}$. If it vanishes, $Z$ can be seen as the pull-back of an ordinary function over $\mathcal{B}/\mathcal{G}$ and is invariant under the action of $\mathcal{G}$.

Our aim will therefore to give a compact formula for these holonomies. In this section, as a preparation to this task,  we describe the space of background fields $\mathcal{B}$, the action of the diffeomorphism and gauge symmetries on $\mathcal{B}$, the anomaly line bundle $\mathscr{A}$ and its natural connection.

\subsection{Abelian gauge fields as differential cocycles}

\label{SecAbGFAsDC}

We now briefly recall how abelian $p$-form gauge fields can be modeled by shifted differential cocycles \cite{hopkins-2005-70, Freed:2006yc}. Let $M$ be a compact orientable manifold. $C^p(M; \mathbbm{K})$ will denote the space of $\mathbbm{K}$-valued cochains of degree $p$ for $\mathbbm{K} = \mathbbm{Z}$, $\mathbbm{R}$ or $\mathbbm{R}/\mathbbm{Z}$ and $Z^p(M; \mathbbm{K})$ the corresponding space of cocycles. Let $\Omega^p(M)$ be the space of smooth real-valued differential forms of degree $p$ on $M$. A shifted degree $p$ differential cochain is a triplet
\be
\check{A} = (f, h, F) \in C^{p}(M; \mathbbm{R}) \times C^{p-1}(M; \mathbbm{R}) \times \Omega^{p}(M) := \check{C}_\bullet^p(M) \;.
\ee
The component $f$ and $F$ are called respectively the characteristic class and the curvature of the differential cochain. Given $c \in Z^{p}(M; \mathbbm{R}/\mathbbm{Z})$, $\check{A}$ belongs to the group $\check{C}_c^p(M)$ of shifted differential cochains shifted by $c$ if $f = c$ modulo 1. Unshifted differential cochains are the cochains shifted by $c = 0$, for which $f \in C^{p}(M; \mathbbm{Z})$. We define a differential on $\check{C}_\bullet^{p}(M)$ by 
\be
\label{EqDefDiffOnDiffCoch}
d\check{A} = (df, F - f - dh, dF) \;,
\ee
where we see $F$ as a real-valued cochain, obtained via integration of $F$ over $p$-chains, in the expression $F - f - dh$. In \eqref{EqDefDiffOnDiffCoch}, we abuse the notation by using $d$ to denote the differential on cochains, on differential forms and on differential cochains. $\check{Z}_c^{p}(M) = \ker d|_{\check{C}_c^{p}(M)}$ is the space of differential cocycles shifted by $c$. We define two shifted differential cocycles to be gauge equivalent if they differ by the differential of an unshifted cochain with vanishing curvature:
\be
\label{EqActGauDiffCoc}
\check{A} = (f,h,F) \sim \check{A} + d\check{B} = (f + dg, h - dj - g, F)
\ee
for $\check{B} = (g,j,0) \in \check{C}_0^{p-1}(M)$. $\check{B}$ is equivalent to $\check{B}' = (g',j',0)$ if $g'-g = da$ and $j'-j = db-a$, for $a \in C^{p-2}(M;\mathbbm{Z})$ and $b \in C^{p-3}(M;\mathbbm{R})$. We write $\check{\mathcal{G}}^p$ for the group generated by these transformations up to equivalence.  

Remark that $\check{\mathcal{G}}^p$ contains non-trivial \emph{micro gauge transformations} \cite{Diaconescu:2003bm}, which act like the identity on $\check{Z}_c^{p}(M)$. Those are given by differential cochains $\check{B}$ with $g$ representing an integral torsion cocycle and $-j = dg$ a real cocycle trivializing $g$ as a real cocycle. These micro gauge transformations will play an important role in the following (see Section \ref{SecGlobTorsAnom}). One can get a more natural and elegant picture by seeing $\check{Z}_c^{p}(M)$ and $\check{\mathcal{G}}^p$ respectively as the objects and arrows of a groupoid, but this point of view will not be required in what follows.

$\check{H}_c^{p}(M) := \check{Z}_c^{p}(M)/\check{\mathcal{G}}^p$ is the space of differential cohomology classes shifted by $c$. Let us denote by $\Omega^p_{\mathbbm{Z}}(M) \subset \Omega^p(M)$ the subgroup of differential forms of degree $p$ with integral periods. $\check{H}_0^{p}(M)$ fits in the following exact sequence:
\be
\label{EqExSeqDiffChar}
0 \rightarrow \Omega^{p-1}(M)/\Omega^{p-1}_{\mathbbm{Z}}(M) \rightarrow \check{H}_0^{p}(M) \rightarrow H^{p}(M;\mathbbm{Z}) \rightarrow 0 \;.
\ee
The differential cohomology classes in $\Omega^{p-1}(M)/\Omega^{p-1}_{\mathbbm{Z}}(M) \subset \check{H}_0^{p}(M)$ are said to be topologically trivial. Such classes always admit representative cocycles of the form $(0,h_A, dA)$, where $A \in \Omega^{p-1}(M)$ and $h_A$ is the real $p$-cocycle associated to $A$ via integration over $(p-1)$-cycles. $\check{H}_c^{p}(M)$ is a torsor on $\check{H}_0^{p}(M)$. From now on, we will write $\check{Z}^{p}(M)$ and $\check{H}^{p}(M)$ instead of $\check{Z}_0^{p}(M)$ and $\check{H}_0^{p}(M)$ for the groups of unshifted differential cocycles and cohomology classes.

Abelian $p$-form gauge fields with ordinary flux quantization laws can be modeled by unshifted differential cocycles of degree $p + 1$. Given $\check{A} = (f,h,F) \in \check{Z}^{p+1}(M)$, we see the curvature $F$ as the field strength of the abelian gauge field. $\chi := \exp 2\pi i h: Z_p(M) \rightarrow U(1)$ is a smooth differential character \cite{springerlink:10.1007/BFb0075216}. It associates an element of $U(1)$ to each smooth $p$-dimensional cycle $\sigma \subset M$, with the property that if $\sigma = \partial \sigma'$, $\chi(\sigma) = \exp 2\pi i \int_{\sigma'} F$. Physically, $\chi$ encodes the Wilson lines of the gauge field around the $p$-cycles of $M$. $f$ is a refinement to integral cohomology of the de Rahm cohomology class of $F$, encoding extra information only on manifolds displaying torsion cohomology. We refer the reader to \cite{Freed:2006yc} for a more detailed and pedagogical exposition of these facts.

Shifted differential cocycles naturally occur when the fluxes of the gauge field $\check{A}$ satisfy a shifted quantization law, like for instance the C-field of M-theory \cite{Witten:1996md}. Indeed, assume temporarily that $M$ is 11-dimensional and let $p_1(M)$ an integral valued cocycle representing the first Pontryagin class of $M$. As $M$ is spin, $p_1(M)$ can be taken to have even periods and $c := \frac{1}{4}p_1(M)$ mod 1 has half-integral periods. Then a C-field is an element of $\check{Z}_c^{p}(M)$, and its gauge equivalence class is an element of $\check{H}_c^{p}(M)$. One can check that the curvatures of elements of $\check{Z}_c^{p}(M)$ do indeed satisfy the quantization law of the field strength of the C-field described in \cite{Witten:1996md}.

\subsection{The space of background fields}

\label{SecSpaBackFields}

The self-dual field displays gravitational anomalies only in dimensions $4\ell+2$ \cite{AlvarezGaume:1983ig}, so we take $M$ to be a compact orientable manifold of dimension $4\ell+2$. Such a manifold can support a self-dual $2\ell$-form unshifted gauge field $\check{T} = (u, t, U)$, whose field strength $U$ satisfies the self-duality condition $U = i\ast U$. 

The self-dual field will be sourced by a (non self-dual) unshifted abelian $2\ell+1$-form gauge field $\check{A} \in \check{Z}^{2\ell+2}(M)$ \cite{Witten:1996hc}. In the case of the self-dual field of type IIB supergravity, this abelian gauge field is a composite of the Ramond and Neveu-Schwarz 2-form gauge fields. In the case of the M5-brane, $\check{A}$ is a certain restriction \cite{Monnierb} of the M-theory C-field 3-form gauge field to the worldvolume of the brane. We refer the reader to \cite{Witten:1996hc} for a more in-depth discussion of the physics and to Section \ref{SecIIB} below for the case of type IIB supergravity. As we will explain later, the background gauge field should be thought of as a differential cocycle shifted by the Wu class, but this shift is trivial in dimension $4\ell+2$ (see Section \ref{SecWuClass}). Genuinely shifted cocycles will appear when $\check{A}$ will be extended to a manifold of dimension $4\ell+4$.

We will also allow for torsion sources, which are required in order to cancel an anomaly \cite{Witten:1999vg, Belov:2006jd}. Recall that the Freed-Witten anomaly \cite{Freed:1999vc} on the worldsheet of strings bound to D-branes requires the characteristic class of the B-field to be a non-vanishing torsion cohomology class on the worldvolume of D-branes which are not Spin$^c$. Although there is no worldsheet picture in our situation, a similar anomaly can plague the self-dual field theory, requiring the characteristic class $f$ of the gauge field $\check{A}$ to be a non-trivial element of $H^{2\ell+2}_{\rm tors}(M; \mathbbm{Z})$ \cite{Witten:1999vg}. We model this fact by adding a fictitious torsion source for the self-dual field, modeled by an unshifted differential cocycle $\check{S} = (s, r, 0) \in \check{Z}^{2\ell+2}(M)$. Here $s$ is an integral cocycle representing a class in $H^{2\ell+2}_{\rm tors}(M; \mathbbm{Z})$. This anomaly will be studied in Section \ref{SecGlobTorsAnom}, where the cohomology class of $s$ will be determined. In equations, we have 
\be
\label{EqCondSources}
d \check{T} = \check{A} - \check{S}\;, \quad d\check{A} = 0 \;, \quad d\check{S} = 0 \;.
\ee
Remark that it would not be possible to satisfy these equations if $\check{A}$ admitted a non-trivial shift.

We are now ready to describe the space $\mathcal{B}$ of background fields for the self-dual field theory, which comprises the metric of $M$ and the $2\ell+1$ abelian gauge field $\check{A} = (f, h, F)$ sourcing the self-dual field. \eqref{EqCondSources} implies $du = f-s$ so we can take $f = s$ by using the gauge symmetry of the self-dual field. This implies that $F$ is exact. Moreover, the equations of motion of $\check{A}$ impose $d^\dagger F = 0$, so we have in fact $F = 0$. 
We can therefore take the degree $2\ell+1$ cocycle $h$ to be a real cocycle $h_A$ determined by a closed differential form $A$ via integration over $2\ell+1$-chains. The space of background gauge fields therefore consists of differential cocycles of the form
\be
\label{EqParGaugeFields}
\check{A} = (s, h_A + r, 0)
\ee
and is effectively parameterized by $\Omega_{\rm closed}^{2\ell+1}(M)$. In the rest of the paper, we sometimes make a slight abuse of notation and write simply $A$ for $h_A$.

In physics terms, we see that the only possible background gauge field is a flat Wilson line, parameterized by the de Rham cohomology class of $A \in \Omega_{\rm closed}^{2\ell+1}(M)$. 

Let us also remark that possible extra sources $\check{S}'$ for the self-dual field have to couple in the same way as $\check{A}$ does. By redefining $\check{A} \rightarrow \check{A} + \check{S}'$, we are back in the situation described by \eqref{EqCondSources}, provided $d\check{S}' = 0$, which we will assume. This situation could occur for instance in the context of the self-dual field on the M5-brane worldvolume, in case open M2-branes end on it \cite{Strominger:1995ac}.

Let $\mathcal{M}$ be the space of Riemannian metrics on $M$. The full space of background fields is
\be
\label{EqSpaceBGFields}
\mathcal{B} = \Omega^{2\ell+1}_{\rm closed}(M) \times \mathcal{M}\;.
\ee

\subsection{The action of the group of local transformations}

\label{SecActGrpLocTrans}

We now proceed to describe the group of local transformations $\mathcal{G}$ and its action on $\mathcal{B}$. Let us consider first the subgroup of diffeomorphisms. We need to consider only the subgroup $\mathcal{G}_{\rm diff} \subset {\rm Diff}(M)$ of diffeomorphisms leaving invariant the data and constraints used to define the theory. Therefore we require the elements of $\mathcal{G}_{\rm diff}$ to leave the differential cohomology class $[\check{S}] \in \check{H}^{2\ell+2}(M)$ invariant. In case extra sources $\check{S}'$ are present, $\mathcal{G}_{\rm diff}$ should leave the differential cocycle $\check{S}'$ invariant. We will see in Sections \ref{SecAnomBundle} and \ref{SecPhysCharSD} (see also \cite{Witten:1996hc, Witten:1999vg, Belov:2006jd}) that an extra piece of data is required in order to define the self-dual field theory, namely a theta characteristic $\tilde{\eta}$. The data encoded by $\tilde{\eta}$ is a symmetric line bundle over the torus $H^{2\ell+1}(M; \mathbbm{R})/\Lambda$, where $\Lambda := H^{2\ell+1}_\mathbbm{Z}(M; \mathbbm{R})$ is the group of de Rham cohomology classes with integral periods. $\mathcal{G}_{\rm diff}$ acts on $H^{2\ell+1}(M; \mathbbm{R})/\Lambda$ by pull-backs and we require this action to preserve $\tilde{\eta}$ as well. The diffeomorphisms outside $\mathcal{G}_{\rm diff}$ should be seen as relating self-dual field theories in non-equivalent backgrounds. Non-trivial phases picked by the partition function under such diffeomorphisms are not obstructions to pushing it down to the space of classes of equivalent backgrounds, and therefore are not anomalies.

Let $\mathcal{G}_{\rm gauge}$ be the group of gauge transformations of $\check{A}$. Naively, one might identify $\mathcal{G}_{\rm gauge}$ with $\mathcal{G}_{\rm naive} := \Omega^{2\ell+1}_{\mathbbm{Z}}(M)$, the group of differential forms with integral periods. $\mathcal{G}_{\rm naive}$ is the extension
\be
0 \rightarrow \Omega^{2\ell+1}_{\rm exact}(M) \rightarrow \mathcal{G}_{\rm naive} \rightarrow H^{2\ell+1}_\mathbbm{Z}(M,\mathbb{R}) \rightarrow 0 \;,
\ee
where $H^{2\ell+1}_\mathbbm{Z}(M,\mathbb{R})$ denotes de Rham cohomology classes with integral periods. The exact forms are generating the gauge transformations connected to the identity and the group of connected components $H^{2\ell+1}_\mathbbm{Z}(M,\mathbb{R})$ classifies the homotopy types of large gauge transformations. However, we saw in Section \ref{SecAbGFAsDC} that the group of gauge transformations is $\mathcal{G}_{\rm gauge} = \check{\mathcal{G}}^{2\ell+2}$. In addition to the transformations above, $\mathcal{G}_{\rm gauge}$ also contains the micro gauge transformations and fits in the exact sequence
\be
0 \rightarrow \Omega^{2\ell+1}_{\rm exact}(M) \rightarrow \mathcal{G}_{\rm gauge} \rightarrow H^{2\ell+1}(M; \mathbbm{Z}) \rightarrow 0 \;.
\ee
The projection $H^{2\ell+1}(M; \mathbbm{Z}) \rightarrow H^{2\ell+1}_\mathbbm{Z}(M; \mathbbm{R})$ induces a projection 
\be
\label{EqProjGGaugeGNaive}
\mathcal{G}_{\rm gauge} \rightarrow \Omega^{2\ell+1}_{\mathbbm{Z}}(M): \check{B} \mapsto B \;.
\ee

Each diffeomorphism $\phi \in \mathcal{G}_{\rm diff}$ acts as an automorphism of $\mathcal{G}_{\rm gauge}$ by pull-back: $\check{B} \rightarrow \phi^\ast \check{B}$ for $\check{B} \in \mathcal{G}_{\rm gauge}$. $\mathcal{G}$ is the semi-direct product $\mathcal{G} = \mathcal{G}_{\rm gauge} \rtimes \mathcal{G}_{\rm diff}$, with multiplication given explicitly by
\be
(\check{B}_2, \phi_2) \cdot (\check{B}_1, \phi_1) = (\phi_2^\ast(\check{B}_1) +  \check{B}_2 , \phi_2 \circ \phi_1) \;, \quad (\check{B}_i, \phi_i) \in \mathcal{G} \,.
\ee
We define the action of $(\check{B},\phi) \in \mathcal{G}$ on $(A, g) \in \mathcal{B}$ by
\be
\label{EqActTransBGFields}
(\check{B},\phi).(A, g) \rightarrow (\phi^\ast A + B, \phi^\ast g) \,.
\ee
Here, $B \in \Omega_{\mathbbm{Z}}^{2\ell+1}(M)$ is the image of $\check{B}$ under \eqref{EqProjGGaugeGNaive}. One can check that the action of the diffeomorphisms on $A$ is consistent with the parameterization \eqref{EqParGaugeFields}:
\be
\phi^\ast \check{A} = (\phi^\ast(s), \phi^\ast(h_A) + \phi^\ast(r),0) \simeq (s, \phi^\ast(h_A) + r, 0) \;,
\ee
where we performed a gauge transformation mapping $\phi^\ast(\check{S})$ back to $\check{S}$. (Such a gauge transformation exists by our hypothesis that $[\check{S}] \in \check{H}^{2\ell+2}(M)$ is invariant under the action of $\mathcal{G}_{\rm diff}$.)

The quotient $\mathcal{F} := \mathcal{B}/\mathcal{G}$ is the space of background fields modulo gauge transformations and diffeomorphisms. We write $\pi_{\mathcal{F}}: \mathcal{B} \rightarrow \mathcal{F}$ for the associated projection. $\mathcal{F}$ is in general singular. For instance if a metric $g$ admits a group $\mathcal{G}_g$ of isometries satisfying the conditions above, $(0,g)$ is a fixed point of $\mathcal{G}_g$. By a line bundle over $\mathcal{F}$, we will always understand a $\mathcal{G}$-equivariant smooth line bundle over $\mathcal{B}$. This definition will allow us to bypass the subtleties associated with the singularities of $\mathcal{F}$. We will also find it useful in the following to consider the quotient $\mathcal{E} = (M \times \mathcal{B})/\mathcal{G}$, where $\mathcal{G}_{\rm diff}$ and $\mathcal{G}_{\rm gauge}$ act respectively by diffeomorphisms and trivially on $M$. Away from the singularities of $\mathcal{F}$, $\mathcal{E}$ is a fiber bundle over $\mathcal{F}$ whose fibers are diffeomorphic to $M$. 

$M \times \mathcal{B}$ admits a canonical horizontal distribution, given by the kernel of the tangent map of the projection on $M$. This distribution pushes down to a distinguished horizontal distribution in the bundle $\mathcal{E} \rightarrow \mathcal{F}$.

Having described the space of background fields, we proceed to construct the anomaly bundle of the self-dual field theory, which is a line bundle over $\mathcal{F}$. To this end, we need a few extra constructions.

\subsection{Relation to modular geometry}

\label{SecRelModGeom}

Consider $\mathcal{B}$ as a trivial fiber bundle over $\mathcal{M}$. Remark that the vertical tangent space at any point is isomorphic to $\Omega_{\rm closed}^{2\ell+1}(M)$.  
A metric $g \in \mathcal{M}$ determines a Hodge star operator, which squares to $-1$ on $\Omega_{\rm closed}^{2\ell+1}(M)$. The vertical tangent space of $\mathcal{B} \rightarrow \mathcal{M}$ also carries an antisymmetric bilinear form, given by the wedge product pairing on $\Omega^{2\ell+1}(M)$. The Hodge star operator restricts to a complex structure on the space of harmonic forms $\mathcal{H}^{2\ell+1}(M; \mathbbm{R}) \subset \Omega^{2\ell+1}(M)$, and the wedge product pairing restricts to a symplectic form. By the Hodge theorem they induce analogous structures on the cohomology group $H^{2\ell+1}(M; \mathbbm{R})$. Let $2n$ be the dimension of $H^{2\ell+1}(M; \mathbbm{R})$. The complex structure on $H^{2\ell+1}(M; \mathbbm{R})$ induced by the metric can be parameterized relative to a given Darboux basis by a complex symmetric $n \times n$ matrix $\tau$ with a positive definite imaginary part (see Section 3.1. of \cite{Monnier2011a} for more details). We write $\mathcal{C}$ for the space of such matrices. As a subspace of a vector space over $\mathbbm{C}$, $\mathcal{C}$ carries a complex structure. As the Hodge star operator was obtained from the metric on $M$, this defines a map $\pi_\mathcal{M}:\mathcal{M} \rightarrow \mathcal{C}$. Combining it with the projection on de Rahm cohomology $\pi_{\rm dR}: \Omega_{\rm closed}^{2\ell+1}(M) \rightarrow H^{2\ell+1}(M; \mathbbm{R})$, we get a map
\be
\pi_\mathcal{B} = \pi_{\rm dR} \times \pi_\mathcal{M}: \mathcal{B} \rightarrow \mathcal{B}_0 := H^{2\ell+1}(M; \mathbbm{R}) \times \mathcal{C} \;.
\ee

As a concrete example, assume that $\ell = 0$ and $M$ is a torus. $\tau$ is then the modular parameter of the torus and $\mathcal{C}$ is the Teichmüller space of the torus, which is the complex upper half-plane. Recall that the Teichmüller space of a torus can be pictured as the space of flat metrics on the torus (of a fixed volume). The map $\pi_\mathcal{M}$ sends a metric on the torus to the unique flat metric in its conformal class, up to homeomorphisms isotopic to the identity.

\subsection{The anomaly line bundle}

\label{SecAnomBundle}

Siegel theta functions $\theta^\eta$ over $\mathcal{B}_0$ are holomorphic functions parameterized by a characteristic $\eta$. $\eta$ can be seen as an element of $H^{2\ell+1}(M; \mathbbm{R})/\Lambda$, where $\Lambda := H^{2\ell+1}_{\mathbbm{Z}}(M; \mathbbm{R})$ is the lattice of de Rahm cohomology classes with integral periods. This identification is however not canonical and relies on a Lagrangian decomposition of $\Lambda$. The group of symplectic transformations of $\Lambda$, $Sp(2n,\mathbbm{Z})$, acts on the characteristics (see Section 3.1 of \cite{Monnier2011a} for explicit formulas). Let $\Gamma_\eta$ be the subgroup of $Sp(2n,\mathbbm{Z})$ leaving fixed $\eta$. The theta function transforms equivariantly under the action of the group $\Lambda \rtimes \Gamma_\eta$, hence it can be seen as the pull-back of the section of a holomorphic line bundle
\be
\mathscr{C}^\eta_{\mathcal{A}} \rightarrow \mathcal{A}^\eta := \mathcal{B}_0/(\Lambda \rtimes \Gamma_\eta) \;,
\ee
the theta line bundle. Remark that $\mathcal{A}^\eta$ is the universal abelian variety over $\mathcal{C}/\Gamma_\eta$.

The group of diffeomorphisms $\mathcal{G}_{\rm diff}$ acts on the space of harmonic $2\ell+1$ forms $\mathcal{H}^{2\ell+1}(M; \mathbbm{R})$ by pull-backs, and this induces an action on the cohomology group $H^{2\ell+1}(M; \mathbbm{R})$. This action preserves the symplectic form and the complex structure, but in general changes the Darboux basis. It factorizes through a subgroup $\Gamma_{\rm diff}$ of $Sp(2n,\mathbbm{Z})$. As the parameterization $\tau$ depends on the Darboux basis, there is an action of $\mathcal{G}_{\rm diff}$ on $\mathcal{C}$, factorizing as well through $\Gamma_{\rm diff}$. In the case of the two dimensional torus, this is the well-known action of the mapping class group $SL(2,\mathbbm{Z})$ on the complex upper half-plane. Let us write $\mathcal{G}_0 := \Lambda \rtimes \Gamma_{{\rm diff}}$. We obtain from $\pi_\mathcal{B}$ a map
\be
\pi: \mathcal{F} \rightarrow  \mathcal{F}_0 := \mathcal{B}_0/\mathcal{G}_0 \;.
\ee
As we required that $\mathcal{G}_{\rm diff}$ leaves the theta characteristic $\tilde{\eta}$ invariant, $\mathcal{G}_0 \subset \Lambda \rtimes \Gamma_{\tilde{\eta}}$ and $\mathcal{F}_0$ is a covering of $\mathcal{A}^{\tilde{\eta}}$. The line bundle $\mathscr{C}^{\tilde{\eta}}_{\mathcal{A}}$ can therefore be pulled back, first to $\mathcal{F}_0$, and then to $\mathcal{F}$. We write $\mathscr{C}^{\tilde{\eta}}$ for the pulled-back bundle over $\mathcal{F}$.

In \cite{Monnier2011} (with some refined arguments presented in Appendix \ref{AppMoreAnBun}), we showed that the anomaly line bundle $\mathscr{A}$ of the self-dual field theory is given by
\be
\label{EqTopClassAnBun}
\mathscr{A} = \mathscr{C}^{\tilde{\eta}} \;.
\ee

\subsection{The connection on the anomaly bundle}

\label{SecConnAnBun}

Having defined the anomaly bundle $\mathscr{A}$, it remains for us to define the anomaly connection $\nabla_\mathscr{A}$. We first describe a natural connection on $\mathscr{C}^\eta_{\mathcal{A}}$.

$\mathscr{C}^\eta_{\mathcal{A}}$ is a holomorphic line bundle over $\mathcal{A}$. It carries a natural hermitian structure, which we describe now%(see Section 3 of \cite{Monnier2011a} for more details)
. Let $\{a^i, b_i\}_{i = 1}^{n}$ be Darboux coordinates on $H^{2\ell+1}(M; \mathbbm{R})$. In complex structure $\tau$, the holomorphic coordinates are $z_i = \tau_{ij}a^j + b_i$. Let us furthermore write $h^{ij} = i((\tau - \bar{\tau})^{-1})^{ij}$. The hermitian structure on $\mathscr{C}^\eta_{\mathcal{A}}$ is given by
\be
\label{EqDefKahlPot}
(s_1, s_2)_{\mathscr{C}^\eta_{\mathcal{A}}} = s_1 \bar{s}_2 \exp -K \;,
\ee 
with
\be
\label{EqKahlPotThetBun}
K = \pi (z_i - \bar{z}_i) h^{ij} (z_j - \bar{z_j}) - \frac{1}{2} \ln \det h \;.
\ee
We will call $K$ the Kähler potential. We derive this formula from the theta multiplier in Appendix \ref{AppDerKahlPot}. For any holomorphic line bundle, there exists a unique holomorphic connection compatible with a given hermitian structure. We can obtain an explicit formula for $\nabla_{\mathscr{C}^\eta_{\mathcal{A}}}$ in the trivialization obtained by pulling back $\mathscr{C}^\eta_{\mathcal{A}}$ to $\mathcal{B}_0$. In terms of $K$, it reads:
\be
\nabla_{\mathscr{C}^\eta_{\mathcal{A}}} = d - \partial K \;,
\ee
where $d$ and $\partial$ are the differential and the Dolbeault operator on $\mathcal{B}_0$. Its curvature can be computed explicitly (see Appendix \ref{AppDerKahlPot}) and is given by
\be
\label{EqCurvThetaBun}
R_{\mathscr{C}^\eta_{\mathcal{A}}} = \partial \bar{\partial} K = -2\pi i da^j \wedge db_j + \frac{1}{2} h^{ik} h^{jl} d\tau_{ij} \wedge d\bar{\tau}_{kl} \;.
\ee
The connection described above pulls back to a connection $\nabla_{\mathscr{C}^\eta}$ on the line bundle $\mathscr{C}^\eta \rightarrow \mathcal{F}$. 

The anomaly bundle $\mathscr{A}$ is isomorphic to $\mathscr{C}^\eta$, but not as a line bundle with connection. The connection $\nabla_{\mathscr{A}}$ on $\mathscr{A}$ differs from $\nabla_{\mathscr{C}^\eta}$ by a one-form $\omega$ pulled back along the map $\mathcal{F} \rightarrow \mathcal{M}/\mathcal{G}_{\rm diff}$. $\omega$ can be described in terms of the Cheeger half-torsion  of the metric \cite{Branson2005}, which is a certain combination of zeta regularized determinants of Laplacians acting on differential forms. The information summarized above about this one-form will be sufficient for our purpose, but the interested reader can find details in \cite{Monnier:2010ww, Monnier2011}.

The curvature of the pullback of $\nabla_{\mathscr{A}}$ to $\mathcal{B}$ is given by
\be
\label{EqCurvAnBun}
R_{\mathscr{A}}  = -2\pi i da^j \wedge db_j + \frac{1}{2} h^{ik} h^{jl} d\tau_{ij} d\bar{\tau}_{kl} + d\omega \;.
\ee
Remark that with respect to the decomposition of $\mathcal{B} = \Omega_{\rm closed}^{2\ell+1}(M) \times \mathcal{M}$, $R_{\mathscr{A}}$ is diagonal. More precisely we have $R_{\mathscr{A}} = R_\Omega + R_\mathcal{M}$, where
\be
\label{EqDefROmega}
R_\Omega = -2\pi i da^j \wedge db_j
\ee
is a two-form pulled back from $\Omega^{2\ell+1}(M)$ and 
\be
R_\mathcal{M} = \frac{1}{2} h^{ik} h^{jl} d\tau_{ij} d\bar{\tau}_{kl} + d\omega
\ee
is a two-form pulled back from $\mathcal{M}$. This property will be crucial in our arguments below.

\subsection{Factorization of paths and holonomies}

\label{SecFactPathHol}

Let $I = [0,1]$. Any path $p : I \rightarrow \mathcal{B}$ can be decomposed into a pair of paths, $p = (p_1, p_2)$ with $p_1 : I \rightarrow \Omega^{2\ell+1}(M)$ and $p_2 : I \rightarrow \mathcal{M}$. The fact that the curvature $R_{\mathscr{A}}$ takes a diagonal form implies that the holonomy of $\nabla_{\mathscr{A}}$ along $p$ is the product of the holonomies along $p_1$ and $p_2$. Moreover, independent reparameterizations of $p_1$ and $p_2$ do not change the holonomies. 

Consider now a loop in $\mathcal{F}$. It can be lifted to an open path $p$ in $\mathcal{B}$, which we can decompose as above: $p = (p_1, p_2)$. The condition that $p$ comes from a loop in $\mathcal{F}$ reads $p(1) = (\check{B},\phi) \cdot p(0)$ for some $(\check{B}, \phi) \in \mathcal{G}$. In terms of $p_1$ and $p_2$, the condition reads
\be
\label{EqCondLoopLift}
p_2(1) = \phi^\ast(p_2(0)) \;, \quad p_1(1) = \phi^\ast(p_1(0)) + B \;.
\ee
We would like to single out two types of loops in $\mathcal{F}$ which will be useful to us in the following. We will call a loop \emph{vertical} if $p_2$ is a constant loop. Vertical loops are contained in a single fiber of $\mathcal{F} \rightarrow \mathcal{M}/\mathcal{G}_{\rm diff}$ and $\phi$ can be taken to be the identity diffeomorphism. We will call a loop \emph{horizontal} if it lifts to a path in $\mathcal{B}$ of the form
\be
\label{EqHorizLoop}
p(t) = \big(t \phi^\ast(A) + (1-t)A, p_2(t) \big) \;,
\ee
where $A \in \Omega^{2\ell+1}(M)$ and $p_2$ is the lift of a loop in $\mathcal{M}/\mathcal{G}_{\rm diff}$. The element $B \in \Omega^{2\ell+1}_{\mathbbm{Z}}(M)$ in \eqref{EqCondLoopLift} vanishes for horizontal loops. Nevertheless paths lifting horizontal loops are in general not constant along the fibers of $\mathcal{F}$. A horizontal loop with vanishing gauge field is a horizontal loop with $A = 0$, and is constant along the fibers.

\begin{proposition}
\label{PropHomDecLoops}
Any loop in $\mathcal{F}$ is homotopic to the composition of a horizontal and a vertical loop.
\end{proposition}

\begin{proof}
Let $p$ be a path in $\mathcal{B}$ lifting a loop in $\mathcal{F}$, satisfying \eqref{EqCondLoopLift}. To prove the proposition, we construct an explicit homotopy between $p$ and the composition of the lifts of a horizontal and a vertical loop. 

Consider the simplex 
\be
S = \{(s,t) \in [0,1] \times [0,1] | s \geq t\} \;. 
\ee
Let $\rho$ be a smooth function on the interior of $S$ tending to $0$ on the boundary components with $s = 1$ and $t = 0$ and to $1$ on the boundary component $s = t$. For instance we can take $\rho(s,t) = \log(s)/\log(t)$. $\rho$ can be continued to a smooth function $\bar{\rho}$ over $S$, except for unavoidable discontinuities at $(0,0)$ and $(1,1)$. 

Consider now the continuous map from $S$ into $\mathcal{B}$ given by $p_S(s,t) = (q(s,t), p_2(s))$, with
\be
q(s,t) = \big(1-\bar{\rho}(s,t)\big)\big(s \phi^\ast(p_1(0)) + (1-s) p_1(0) + t B\big) + \bar{\rho}(s,t) p_1(s) \;.
\ee
\eqref{EqCondLoopLift} implies that although $\bar{\rho}$ is not continuous at $(0,0)$ and $(1,1)$, $q$ and $p_S$ are. But 
\be
p_S(s,0) = \big(s \phi^\ast(p_1(0)) + (1-s) p_1(0), p_2(s)\big)
\ee 
is the lift of a horizontal loop, 
\be
p_S(1,t) = \big(\phi^\ast(p_1(0)) + tB, p_2(1))
\ee 
is the lift of a vertical loop and $p_S(t,t)$ coincides with $p(t)$. 
\end{proof}

\section{A formula for the holonomies of $\nabla_\mathscr{A}$}

\label{SecHSForm}

In this section, we describe a holonomy formula associating to each loop in $\mathcal{F}$ an element of $U(1)$. We start in Section \ref{SecWuClass} by reviewing the Wu class and defining differential cocycles shifted by the Wu class. In Section \ref{SecFormLiftAdLim}, we introduce the adiabatic limit for fiber bundles with base $B \subset \mathcal{B}$ and discuss the behavior of characteristic forms associated to the metric in this limit. In Section \ref{SecLoopMapTor}, we show how to associate to any loop $c$ in $\mathcal{F}$ a mapping torus $M_c$ endowed with a metric and a shifted differential cocycle. We then extend this data to a manifold $W$ bounded by the mapping torus in Section \ref{SecBoundMan}. This allows us to define the holonomy formula in Section \ref{SecHolForm} as the integral of a certain top form on $W$. We show that the holonomy formula is independent of the arbitrary choices involved in the construction of $W$ and the top form. We show in Section \ref{SecGlobTorsAnom} that the torsion anomaly identified in \cite{Witten:1999vg} is correctly reproduced in this formalism. We also explain in Section \ref{SecPhysCharSD} how to determine the theta characteristic $\tilde{\eta}$, which appeared in the construction of the anomaly line bundle in Section \ref{SecAnomBundle}.

\subsection{Differential cocycles shifted by the Wu class}

\label{SecWuClass}

In this section, $X$ is a manifold of arbitrary dimension $d$. Recall \cite{opac-b1077949} that the Wu class $\nu \in H^\bullet(X; \mathbbm{Z}_2)$ is a characteristic class which can be expressed in terms of the Stiefel-Whitney classes on $X$. Writing $\nu_k$ for its degree $k$ component, the Wu class has the property that
\be
\nu_k \cup x = {\rm Sq}^k(x) 
\ee
for all $x \in H^{d-k}(X; \mathbbm{Z}_2)$, where Sq is the Steenrod square operation. In particular, $\nu_k = 0$ on a manifold of dimension $d$ if $k > d/2$, and if $d$ is even,
\be
\label{EqPropWuClassMidDeg}
\nu_{d/2} \cup x = x \cup x 
\ee
for $x \in H^{d/2}(X; \mathbbm{Z}_2)$.

Under the inclusion $\mathbbm{Z}_2 \subset \mathbbm{R}/\mathbbm{Z}$, any $\mathbbm{Z}_2$-cocycle representative of $\nu_p$ determines an element $\lambda \in Z^p(X;\mathbbm{R}/\mathbbm{Z})$. We say that $\check{Z}_{\lambda}^p(X)$ is the space of differential cocycles shifted by the Wu class. $\check{H}_{\lambda}^p(X)$ does not depend on the choice of representing cocycle.

Remark that if $\check{A}= (f,h,F) \in \check{Z}_{\lambda}^p(X)$, then $2F$ has integral periods. If $C$ is a cycle, $\int_C 2F$ is even if $\langle \nu, [C] \rangle = 0$ and odd if $\langle \nu, [C] \rangle = 1$, where the brackets denote the pairing between homology and cohomology. We will call \emph{form lifts of the Wu class} the closed differential forms with integral periods sharing this property. Similarly, we call \emph{integral lifts of the Wu class} the classes in $H^p(X;\mathbbm{Z})$ reducing modulo 2 to the Wu class. (Assuming that $\nu$ is in the image of the reduction modulo 2 map.) 
If the dimension $d$ of $X$ is even and $p = d/2$, the compatibility of the wedge product pairing on closed forms and of the cup product pairing on integral cohomology, combined with \eqref{EqPropWuClassMidDeg}, implies that $2F$ is a characteristic element for the wedge product pairing on the space of forms with integral periods; namely
\be
\label{EqDefCharEl}
\int_X G \wedge G = \int_X G \wedge 2F \quad {\rm mod} \; 2
\ee
for any closed form with integral periods $G \in \Omega^{d/2}_\mathbbm{Z}(X)$. Remark that for manifolds of dimension less than $2p$, $\nu_p$ vanishes. In this case, $\check{Z}^p_\lambda(X) = \check{Z}^p(X)$ and $F$ has integral periods.

From now on, we go back to considering a manifold $M$ of dimension $4\ell+2$. A \emph{lift of the Wu class} will always mean a lift of the Wu class of degree $2\ell+2$.

\subsection{The adiabatic limit}

\label{SecFormLiftAdLim}

Recall that $\mathcal{E} = (M \times \mathcal{B})/\mathcal{G}$ is a fiber bundle over $\mathcal{F}$ whose fibers are diffeomorphic to $M$. Let $B \subset \mathcal{F}$ be a finite-dimensional smooth submanifold. We can restrict $\mathcal{E}$ to $B$ and obtain fiber bundle $E \rightarrow B$ with fiber $M$. The inclusion of the fiber over the point $b \in B$ will always be denoted by $i_b: M \rightarrow E$, and the image of $i_b$ will be written $M_b$. The horizontal distribution in $\mathcal{E}$ restricts to a horizontal distribution in $E$. Any tangent vector in $E$ decomposes uniquely into a horizontal part, belonging to the horizontal distribution, and a vertical part, which is tangent to a fiber. Dually, this induces vertical and horizontal subspaces in $\Omega^\bullet(E)$. A \emph{horizontal form} is the pull-back of a form on $B$, while a \emph{vertical form} vanishes when contracted with at least one horizontal vector. 

The projection $E \rightarrow B$ determines a metric on each fiber of $E$, i.e. a metric $g_V$ on the vertical tangent bundle $T_VE$. Pick an arbitrary smooth metric $g_B$ on $B$, and construct the metric
\be
\label{EqMetFibInF}
g_E = g_V \oplus g_B/\epsilon  \;.
\ee 
for $\epsilon \in \mathbbm{R}_+$. The direct sum above refers to the decomposition $TE$ into vertical and horizontal vectors. $g_E$ obviously restricts to $g_V$ on the vertical tangent bundle. 

We now consider the \emph{adiabatic limit} \cite{Witten:1985xe} $\epsilon \rightarrow 0$. In this limit, the size of $B$ diverges. We have 
\begin{proposition}
\label{PropCharFormIndepgB}
Assume $g_B$ is flat. Then the characteristic forms constructed from the metric $g_E$ are independent of $g_B$ in the adiabatic limit. 
\end{proposition}
\begin{proof}
Let $R_V$ be the curvature of the Levi-Civita connection associated to the vertical metric and let $R_B$ be the curvature of the Levi-Civita connection associated to $g_B$. In Section 3.j of \cite{MR861886}, it is shown that the curvature of the Levi-Civita connection of $g_E$ has the following asymptopic expansion in the adiabatic limit:
\be
\label{EqBFCurvAdiabLim}
R_E = \left( \begin{array}{cc} R_V & 0 \\ 0 & R_B \end{array} \right) + O(\epsilon^{1/2}) \;.
\ee
Recall that the Riemannian curvature is a 2-form valued in the endomorphisms of the tangent bundle. In equation \eqref{EqBFCurvAdiabLim}, we pictured $R_E$, $R_V$ and $R_B$ as matrices of 2-forms, with the two blocks corresponding to the decomposition of $TE$ into vertical and horizontal vectors.

A characteristic form depends on $R_E$ through the elementary monomials ${\rm Tr}(R_E^p)$. If $g_B$ is flat as is assumed, then \eqref{EqBFCurvAdiabLim} makes it clear that ${\rm Tr}(R_E^p)$ is independent of $g_B$ up to terms of order $O(\epsilon^{1/2})$.  
\end{proof}

\subsection{Loops and mapping tori}

\label{SecLoopMapTor}

Let us specialize the construction of the fiber bundle $E$ of Section \ref{SecFormLiftAdLim} to the case where the base $B$ is a loop $c$ in $\mathcal{F}$, and call it $M_c$. $M_c$ is a fiber bundle over $S^1$ with fiber $M$, called the mapping torus associated to $c$. We pick an arbitrary metric on $S^1$ and endow $M_c$ with the $\epsilon$-dependent metric \eqref{EqMetFibInF}, written $g_c$.

Recall that in addition to a metric on $M$, points in $\mathcal{F}$ parameterize a flat differential cocycle of degree $2\ell+2$ on $M$, modulo the action of $\mathcal{G}$. We would like to construct on $M_c$ a differential cocycle that restrict to the corresponding data on each fiber.

\begin{lemma}
\label{LemConstrDiffCocMapTor}

Let $\check{A}_t = (s, h_{A_t} + r, 0)$, $t \in [0,1]$ be a family of differential $2\ell+2$-cocycles on $M$, and let $A_t$ be the associated family $2\ell+1$-forms. Assume that 
\be
\check{A}_1 = \phi^\ast(\check{A}_0) + d(\check{B} + \check{V}) \;,
\ee
where $(\check{B},\phi) \in \mathcal{G}$ and $d\check{V} = \phi^\ast(\check{S}) - \check{S}$. Then there is a differential cocycle $\check{A}_c \in \check{H}^{2\ell+2}(M_c)$ such that
\be
\label{EqRestFibDiffCocMapTor}
i^\ast_t (\check{A}_c) = \check{A}_t \;.
\ee
\end{lemma}
\begin{proof}
Write $\check{B} = (g, j, 0)$, $\check{S} = (s, r, 0)$ and $\check{V} = (u,v,0)$. The 1-parameter family $h_{A_t} + r - t(dj + dv)$ defines a unique degree $2\ell+1$ real cochain $h_{M \times I}$ on $M \times I$. It pushes down to a cochain $h_c$ on $M_c$, which is discontinuous at $M_0 \simeq M_1$. The discontinuity is given by $\Delta h = g + u \in \check{Z}^{2\ell+1}(M;\mathbbm{Z})$.

Let $a$ be an integral $2\ell+2$-chain in $M_c$ and $a_0$ its restriction to $M_0$. Let $f_\Delta$ be the integral $2\ell+2$-cochain on $M_c$ taking value $\Delta h(a_0)$ on $a$. $f_\Delta$ is a cocycle. We define $f_c = f_\Delta + s \in Z^{2\ell+2}(M;\mathbbm{Z})$, where as before, we see $s$ as a cochain on $M \times I$ and push it down to a cochain on $M_c$. Remark that by construction, as $h_c$ jumps by $\Delta h$ across $M_0$, $f_c + dh_c$ is a smooth cocycle.

For definiteness, we assume that the path is parameterized such that it is constant in $[0,\delta]$ and $[1-\delta,1]$ for a small $\delta \in \mathbbm{R}_+$. We then define the form $F_{M \times I} = dt \wedge \partial_t A_t$ on $M \times I$, where the 1-parameter family of forms $\partial_t A_t$ on $M$ is seen as a single vertical form on $M \times I$. $F_{M \times I}$ is the pull-back of a unique smooth form $F_c$ on $M_c$. The relations 
\be
df_c = 0 \;, \quad F_c - f_c - dh_c = 0 \quad {\rm and} \quad dF_c = 0
\ee
hold, so 
\be
\check{A}_c = (f_c, h_c, F_c)
\ee
is a differential cocycle on $M_c$. Moreover, recalling that $i_t$ denotes the inclusion of the fiber over $t \in S^1$, we obviously have $i^\ast_t(\check{A}_c) = \check{A}_t$.
\end{proof}

\begin{corollary}
\label{CorConstrDiffCocMapTor}
Given a loop $c \subset \mathcal{F}$ lifting to a path $p \subset \mathcal{B}$, we get a one parameter family $A_t$ of $2\ell+1$-forms on $M$, satisfying 
\be
A_1 = \phi^\ast(A_0) + B \;,
\ee
where $B \in \Omega_{\mathbbm{Z}}^{2\ell+1}(M)$. Any lift of $B$ to a differential cocycle $\check{B} \in \mathcal{G}_{\rm gauge}$ determines a differential cocycle $\check{A}_c$ on $M_c$ satisfying \eqref{EqRestFibDiffCocMapTor}.
\end{corollary}

In conclusion, we can associate with the loop $c \subset \mathcal{F}$ a mapping torus $M_c$ endowed with a one-parameter family of metrics $g_c$ as well as an unshifted differential cocycle $\check{A}_c$. It should be emphasized that this assignment is not unique. We made arbitrary choices when we picked a path $p$ lifting $c$ to $\mathcal{B}$, when we picked a metric $g_c$ on $S^1$, and when we lifted $B \in \Omega^{2\ell+1}_{\mathbbm{Z}}(M)$ to a differential cocycle $\check{B}$. Remark that we can see $\check{A}_c$ as a differential cocycle shifted by the Wu class, because $\nu_{2\ell+2}$ vanishes on a $4\ell+3$-dimensional manifold.

\subsection{Bounded manifold}

\label{SecBoundMan}

Assume that there exists a manifold $W$ admitting $M_c$ as its boundary. The existence of $W$ depends on the vanishing of a certain cobordism group. For spin manifolds of dimension $3$, $7$ and $11$, relevant for most applications to physics, the corresponding cobordism groups were shown to vanish in \cite{2012arXiv1208.1540M} (see also \cite{Kriz:2004tu} for the case of dimension $11$). Among the physically relevant self-dual fields, only the one on the string theory and M-theory five-branes can live on manifolds which are non-spin. The corresponding cobordism group is described in \cite{Monnierb}, but it is unknown if it vanishes.

We endow $W$ with a 1-parameter family of metrics $g_W$ on $W$, with parameter $\epsilon$, reducing to the metric $g_c$ on $M_c = \partial W$. We will suppress the dependence on $\epsilon$ in the notation. We require in addition that there is a neighborhood of $M_c$ in $W$ on which $g_W$ is the direct product metric.

Let $j$ be the inclusion of $M_c$ as the boundary of $W$. The fact that $M_c$ has trivial normal bundle in $W$ implies that the Wu classes $\nu_{M_c}$ and $\nu_W$ of $M_c$ and $W$ are related by $j^\ast(\nu_W) = \nu_{M_c} = 0$. Therefore there is a well-defined restriction map $j^\ast: \check{Z}^{2\ell+2}_\nu(W) \rightarrow \check{Z}^{2\ell+2}(M_c)$. We endow $W$ with a shifted differential cocycle $\check{A}_W \in \check{Z}^{2\ell+2}_\nu(W)$ such that $j^\ast \check{A}_W = \check{A}_c$. Again, the existence of this cocycle is linked to the vanishing of certain cobordism groups, which were shown to vanish for spin manifolds of physical relevance in \cite{2012arXiv1208.1540M}. The corresponding statement for five-branes is still open.

We obtain in this way a manifold $W$ bounded by $M_c$, endowed with a metric and a shifted differential cocycle extending the corresponding structures on $M_c$.

\subsection{The holonomy formula}

\label{SecHolForm}

We associate to a cycle $c \subset \mathcal{F}$ a complex number ${\rm hol}(c)$ of unit modulus, given by
\be
\label{EqHolForm}
\frac{1}{2\pi i} \ln h(c) = \lim_{\epsilon \rightarrow 0} \frac{1}{8} \int_W \left(L(TW) - 4 F_W^2 \right) \;.
\ee
In the formula above, $W$ is as described in Section \ref{SecBoundMan}.
$L(TW)$ is the Hirzebruch genus of the tangent bundle of $W$, constructed from the Pontryagin forms associated to the metric $g_{W}$ and $F_W$ is the curvature of the shifted differential cocycle $\check{A}_W$. We claim that 
\be
\label{EqHolFormCl}
h(c) = {\rm hol}_\mathscr{A}(c) \;,
\ee
where ${\rm hol}_\mathscr{A}(c)$ is the holonomy of the anomaly connection $\nabla_{\mathscr{A}}$ along the loop $c \subset \mathcal{F}$.

In \cite{Monnier2011a}, we gave strong evidence for \eqref{EqHolFormCl}, in the case of loops in $\mathcal{F}$ lifting to paths of the form $(0,p_2) \in \mathcal{B}$. (Section 5.1 of \cite{Monnier2011a} lists the reasons why these arguments fall short of a proof.) In Section \ref{SecProofAnForm}, we will show that if \eqref{EqHolFormCl} holds for such loops, then it holds for all loops in $\mathcal{F}$.

We can interpret \eqref{EqHolForm} as the evaluation on $M_c$ of a Chern-Simons functional associated to the characteristic form $\frac{1}{8} \left(L(TW) - 4 F_W^2\right)$. (The geometrical structure of which $F_W$ originates as a characteristic form is not completely clear, see for instance the discussion in the introduction of \cite{hopkins-2005-70}.) We see therefore that  the claim \eqref{EqHolFormCl} identifies $\mathscr{A}$ with its natural connection with the corresponding Chern-Simons line bundle (in the sense of \cite{Freed:1992vw}). Such Chern-Simons functionals were studied in \cite{hopkins-2005-70} as cobordism invariants for families of manifolds. No adiabatic limit is taken in the setting of \cite{hopkins-2005-70}, however.

Let us recall that several arbitrary choices are involved in the definition of \eqref{EqHolForm}. In order to define a metric and a differential cocycle on the mapping torus, we first had to pick a path $p$ in $\mathcal{B}$ lifting the loop $c \in \mathcal{F}$. We then had to pick a metric $g_{S^1}$ on $S^1$ to construct the metric $g_c$ on $M_c$. Also, the manifold $W$ and the extension of the metric and the differential cocycle on it are not unique. We have to show that $h(c)$ is independent of all those choices. 

\begin{proposition}
\label{PropIndepBoundMan}
$h(c)$ is independent of the choice of bounded manifold $W$.
\end{proposition}
\begin{proof}
Consider two manifolds $W_1$ and $W_2$ bounded by $M_c$ and respectively endowed by shifted differential cocycles  $\check{A}_{W_i} = (f_i, h_{W_i}, F_{W_i})$, $i = 1,2$, extending $\check{A}_c$. Let $\bar{W}_2$ be $W_2$ with its orientation reversed. Then we can glue $W_1$ to $\bar{W}_2$ smoothly along their common boundary $M_c$ to obtain a closed manifold $W$. The metric on $W$ determines Pontryagin forms, from which we obtain the Hirzebruch genus $L(TW)$. Moreover, the differential cocycles $\check{A}_{W_1}$ and $\check{A}_{W_2}$ glue smoothly at the boundary, yielding a differential cocycle $\check{A}_W = (f_W, h_W, F_W)$ on $W$. We have
\be
\label{EqHolIndepBoundedM}
\frac{1}{8} \int_W \left(L(TW) - 4 F_W^2 \right) = \frac{1}{8} \int_{W_1} \left(L(TW_1) - 4 F_{W_1}^2 \right) - \frac{1}{8} \int_{W_2} \left(L(TW_2) - 4 F_{W_2}^2 \right) \;.
\ee
To complete the proof of the proposition, we have to show that the left-hand side is an integer.

To this end, recall that $2F_W$ is a characteristic element for the wedge product pairing modulo 2, i.e. it satisfies \eqref{EqDefCharEl}. Given a lattice with an integral unimodular pairing, the norm square of a characteristic element is always equal to the signature of the pairing modulo 8 (see for instance Theorem 2.8 of \cite{Brumfiel1973} for a more general statement including non-unimodular pairings). The wedge product pairing on a closed manifold is unimodular, and the integral of the Hirzebruch genus over $W$ gives precisely the signature of the wedge product pairing. In consequence, the left-hand side of \eqref{EqHolIndepBoundedM} is an integer indeed.
\end{proof}

\begin{proposition}
\label{PropHolIndepMetS1}
$h(c)$ is independent of the metric $g_{S^1}$ involved in the construction of $g_c$.
\end{proposition}
\begin{proof}
Let us see $g_c$ as a function of $g_{S^1}$. By making some definite choice of extension, we can see $g_W$ as a function of $g_{S^1}$ as well. $F_W$ is independent of $g_W$. The Hirzebruch genus $L(TW)$ is expressible in terms of the Pontryagin forms of $TW$. As these forms are closed and vary by exact forms when $g_W$ is modified, the integral in \eqref{EqHolForm} can depend on $g_W$ only through its boundary values, $g_c$. But Proposition \ref{PropCharFormIndepgB} shows that the characteristic forms of $TM_c$ are independent of $g_{S^1}$ in the adiabatic limit $\epsilon \rightarrow 0$. 

\end{proof}

\begin{proposition}
\label{PropHolIndepPathLift}
$h(c)$ is independent of the path lift $p$.
\end{proposition}
\begin{proof}
Consider two paths $p_0$ and $p_1$ in $\mathcal{B}$ such that $\pi_{\mathcal{F}}(p_0(t)) = \pi_{\mathcal{F}}(p_1(t)) = c(t)$. Together with a choice of metric on $S^1$, these two paths define metrics and differential cocycles on two copies of the mapping torus $M_c$. We see these two copies as the two boundary components of $W = M_c \times I$, $I = [0,1]$. In order to prove the proposition, we will extend the metric and the differential cocycle from $\partial W$ to $W$ and show that \eqref{EqHolForm} evaluated on this data yields an integer. A simple gluing argument then shows that this integer is nothing but the difference between logarithms of the holonomies computed from $p_0$ and $p_1$, thereby showing that the holonomies coincide.

Let $p_s$, $s \in I$ be a smooth family of paths in $\mathcal{B}$ interpolating between $p_0$ and $p_1$ and satisfying as well $\pi_{\mathcal{F}}(p_s(t)) = c(t)$. $p_s$ defines a two-parameter family of metrics and differential cocycles $(\check{A}_{t,s},g_{t,s})$ on the fibers of $W \rightarrow S^1 \times I$. We pick a metric $g_{S^1}$ on $S^1$ and the Euclidean metric $g_I$ on $I$ to construct the metric $g_W =  g_{t,s} \oplus (g_c \oplus g_I)/\epsilon$ on $W$.

Let us write $M_{c,s}$ for the slice $M_c \times \{s\} \subset W$, with its induced metric. Similarly, let $M_{t,s} = M \times \{t\} \times \{s\}$. As $H^{4\ell+4}(W; \mathbbm{R}) = 0$, the top component of the Hirzebruch genus is exact, $L(TW) = dK$. Moreover, the metrics on $M_{c,0}$ and $M_{c,1}$ are related by a diffeomorphism $\phi$. This implies we can take $K$ such that $K|_{M_{c,1}} = \phi^\ast K|_{M_{c,0}}$, so the boundary contribution to the integral of the Hirzebruch genus in \eqref{EqHolForm} cancel.

We can follow the same logic as in Section \ref{SecLoopMapTor} and construct a differential cocycle $\check{A}_W = (f_W, h_W, F_W)$ on $W$ out of the two parameter family of differential cocycles $\check{A}_{t,s}$. What remains to be computed is the integral over $W$ of $F_W^2$. For this purpose however, it is simpler to cut open the mapping tori in $W$ and define $\tilde{W} = M \times I \times I$. The map $\tilde{W} \rightarrow W$ allows us to pull-back $\check{A}_W$ to a differential cocycle $\check{A}_{\tilde{W}} = (0, A_{t,s}, F_{\tilde{W}})$.

We now compute the integral of $F^2_{\tilde{W}}$. $F_{\tilde{W}}$ (written henceforth $F$) vanishes when restricted to a fiber: $F|_{M_{t,s}} = 0$. We can therefore write $F = dt \wedge ds \wedge F^{\rm ts} + dt \wedge F^{\rm t} + ds \wedge F^{\rm s}$, where $F^{\rm ts}$, $F^{\rm t}$ and $F^{\rm s}$ are closed vertical forms of degree $2\ell$, $2\ell+1$ and $2\ell+1$ on $W$, respectively. We can rewrite
\be
4\int_{\tilde{W}} F^2 = - 4 \int_I dt \int_I ds \int_{M} F^{\rm t} \wedge F^{\rm s} \;.
\ee
But 
\be
F^{\rm s} = \partial_s A_{t,s}
\ee
is exact, because $A$ changes by a gauge transformation along $s$. This implies that $\int_{\tilde{W}} F^2$ vanishes.

\end{proof}

\subsection{The torsion anomaly}

\label{SecGlobTorsAnom}

There is still a rather subtle way in which \eqref{EqHolForm} might be ill-defined. Recall that in Corollary \ref{CorConstrDiffCocMapTor}, we had to lift the element $B \in \Omega^{2\ell+1}_{\mathbbm{Z}}(M)$ determined by the path lift of $c$ to a differential cocycle $\check{B}$, in order to construct the differential cocycle $\check{A}_c$ on the mapping torus. There is no unique way to do this. In fact, the ambiguity is parameterized by $H^{2\ell+1}_{\rm tors}(M; \mathbbm{Z})$, for reasons discussed in Section \ref{SecActGrpLocTrans}. We have to make sure that $h(c)$ is independent of the choice of lift. It turns out that this is not necessarily the case, and that the obstruction coincides with the torsion anomaly identified in \cite{Witten:1999vg, Belov:2006jd}.

To investigate the dependence on the choice of lift of $B$, we study \eqref{EqHolForm} on constant loops. Therefore assume that $c$ lifts to a constant path in $\mathcal{B}$. $M_c$ is then a trivial mapping torus, isometric to $M \times S^1$. In order to construct the differential cocycle $\check{A}_c$ on $M_c$, according to Corollary \ref{CorConstrDiffCocMapTor} we have to pick a differential cocycle $\check{B}$ lifting $B = 0 \in \Omega^{2\ell+1}_{\mathbbm{Z}}(M)$. A generic lift has the form $\check{B} = (u, v, 0)$, where $u \in Z^{2\ell+1}(M; \mathbbm{Z})$ represents a class in $H^{2\ell+1}_{\rm tors}(M; \mathbbm{Z})$ and $v$ is a real cochain such that $u = -dv$. Corollary \ref{CorConstrDiffCocMapTor} produces a cocycle $\check{A}_c$ on $M_c$. We extend $\check{A}_c$ to a shifted differential cocycle $\check{A}_W = (f_W, h_W, F_W)$ on $W$. In general, the curvature $F_W \in \Omega^{2\ell+2}(W)$ is non-vanishing, although of course, $F_W|_{M_c} = 0$. 

Consider $W_0 = M \times D^2$, where $D^2$ denotes a disk. We glue $W$ to $W_0$ with its orientation reversed to obtain a closed manifold $W_C := W \ltxcup \bar{W}_0$. The Wu class of degree $2\ell+2$ vanishes on $W_C$, so we can pick an integral lift $\lambda_C \in H^{2\ell+2}(W_C;\mathbbm{Z})$ of the Wu class such that $\lambda_C|_{\bar{W}_0} = 0$. $\lambda := \lambda_C|_{W} \in H^{2\ell+2}(W,M_c;\mathbbm{Z})$ is an integral lift of the Wu class of $W$. Importantly for what follows, not all integral lifts of the Wu class on $W$ can be obtained in this way; this is the reason for the construction above. $2f_W$ is an integral cocycle, so let us pick $\tilde{f}_W \in H^{2\ell+2}(W;\mathbbm{Z})$ such that $2\tilde{f}_W = 2f_W - \lambda$, $\tilde{f}_W|_{M_c} = f_W|_{M_c}$ in cohomology. Recalling that $s$ is the torsion characteristic class of $\check{A}$ on $M$, let us write $\exp 2\pi i a_{s,u}$ for \eqref{EqHolForm} evaluated with the data above. We see that
\begin{align}
\label{EqHolMicroGaugeTrans}
a_{s,u} = \, & \frac{1}{8} \lim_{\epsilon \rightarrow 0} \int_W \left(L(TW) - 4f_W^2 \right) \\
= \, & \frac{1}{8} \lim_{\epsilon \rightarrow 0} \left( \int_W L(TW) - \sigma \right) + \frac{1}{8}\left(\sigma  - \int_W \lambda^2\right) - \frac{1}{2}\int_W \left(\tilde{f}_W (\tilde{f}_W + \lambda) \right) \;, \notag
\end{align}
where $\sigma$ is the signature of the pairing on the relative cohomology $H^{2\ell+2}(W,M_c;\mathbbm{R})$. We used the fact that $F_W|_{M_c} = 0$ to reexpress its integral over $W$ as the pairing of the corresponding cohomology class with the fundamental homology class of $(W,\partial W)$, which we still write as an integral. Our aim will be now to compute more explicitly $a_{s,u}$. 

The Atiyah-Patodi-Singer theorem \cite{Atiyah1973} states that $\eta := \int_W L(TW) - \sigma$ is the eta invariant of the signature Dirac operator on $M_c$. It is easy to check that $\eta$ vanishes by applying again the Atiyah-Patodi-Singer theorem on $W' = M \times D^2$ with a direct product metric extending the metric on $M_c = M \times S^1$. The first term in  \eqref{EqHolMicroGaugeTrans} is therefore zero. 

We need more background to understand the last two terms of \eqref{EqHolMicroGaugeTrans}. Most of the following material can be found in \cite{Brumfiel1973}. Recall that for $X = \emptyset$ or $\partial W$, we have short exact sequences
\be
0 \rightarrow H^p_{\rm tors}(W,X;\mathbbm{Z}) \rightarrow H^p(W,X;\mathbbm{Z}) \rightarrow H^p_{\rm free}(W,X;\mathbbm{Z}) \rightarrow 0 \;.
\ee
$H^p_{\rm free}(W,X;\mathbbm{Z})$ coincides with the lattice in $H^p(W,X;\mathbbm{R})$ consisting of classes with integral periods. The relative cohomology groups themselves fit in a long exact sequence of cohomology groups
\be
... \rightarrow H^{2\ell+1}(\partial W; \mathbbm{Z}) \rightarrow H^{2\ell+2}(W, \partial W; \mathbbm{Z}) \stackrel{\iota}{\rightarrow} H^{2\ell+2}(W; \mathbbm{Z}) \stackrel{\rho}{\rightarrow} H^{2\ell+2}(\partial W; \mathbbm{Z}) \rightarrow ... \;.
\ee
This long exact sequence induces analogous long sequences for the torsion and free cohomology groups, which are however not exact. We will call $\iota_t$, $\rho_t$, $\iota_f$ and $\rho_f$ the analogues of $\iota$ and $\rho$ for the torsion and free cohomology groups, respectively.

Consider the finite group $G = {\rm ker} \rho_f / {\rm img} \iota_f$. The cup product pairing modulo 1 on $G$ is a symmetric non-degenerate pairing, which we will write $L_G$. By the definition of the Wu class, $\lambda$ is a characteristic element for the cup product pairing on $H^{2\ell+2}(W, \partial W; \mathbbm{Z})$, i.e.
\be
\int_W w \cup \lambda = \int_W w^2 \quad {\rm mod} \; 2
\ee
for all $w \in H^{2\ell+2}(W, \partial W; \mathbbm{Z})$. Theorem 2.4 of \cite{Brumfiel1973} shows that 
\be
\label{EqDefQROnW}
q_W(z) = \frac{1}{2}\int_W z \cup (z + \lambda) \quad {\rm mod} \; 1
\ee
for $z \in {\rm ker} \rho_f$ is a quadratic refinement of $L_G$, i.e. a $\mathbbm{R}/\mathbbm{Z}$-valued function on $G$ satisfying
\be
\label{EqDefQR}
q(z_1 + z_2) = q(z_1) + q(z_2) + L_G(z_1,z_2) \;, \quad q(nz_1) = n^2 q(z_1) \;.
\ee
A thorough presentation of the properties of quadratic refinements can be found in \cite{Taylor}. This characterizes the third term of \eqref{EqHolMicroGaugeTrans}, which is equal to $-q_W(\tilde{f}_W)$. 

Given a quadratic refinement $q$ on a finite group $G$, its Arf invariant is the $\mathbbm{R}/\mathbbm{Z}$-valued complex argument of the associated Gauss sum:
\be
{\rm Gauss}(q) = \sum_{z \in G} \exp 2 \pi i q(z) =: |{\rm Gauss}(q)| \exp 2 \pi i {\rm Arf}(q)
\ee
Theorem 2.8 of \cite{Brumfiel1973} shows that the second term in \eqref{EqHolMicroGaugeTrans} coincides with the Arf invariant of $q_W$.

We now relate $q_W$ to data on $\partial W$. If $z_1,z_2 \in H^{2\ell+2}(W;\mathbbm{Z})$ restrict to torsion classes $y_1$ and $y_2$ on $\partial W$, then Lemma 3.4 of \cite{Brumfiel1973} shows that
\be
L_G(z_1,z_2) = -L_T(y_1,y_2) \quad {\rm mod} \; 1 \;,
\ee
where $L_T$ denotes the linking pairing on $H^{2\ell+2}_{\rm tors}(\partial W;\mathbbm{Z})$. In consequence, we see that $q_W$ determines a quadratic refinement $q_c$ of $L_T$ by
\be
q_c(f_c) = -q_W(f_W) \quad {\rm mod} \; 1 \;.
\ee 
Now recall that $\partial W = M_c = M \times S^1$ and that $f_c = f_W|_{M_c} = s + u \cup t$, where $t$ is the pull-back of the generator of $H^{1}(S^1;\mathbbm{Z})$. We can therefore characterize the anomaly by
\be
a_{s,u} = q_c(s + u \cup t) - {\rm Arf}(q_c) \;,
\ee
where ${\rm Arf}()$ denotes the Arf invariant.

Let us now study the properties of $\mathbbm{Z}_2$-valued quadratic refinements of $L_T$ in more detail. Remark first that we have a canonical isomorphism
\be
\label{EqDecompTorsCohomMS1}
H^{2\ell+2}_{\rm tors}(M \times S^1;\mathbbm{Z}) = H^{2\ell+1}_{\rm tors}(M; \mathbbm{Z}) \cup t + H^{2\ell+2}_{\rm tors}(M; \mathbbm{Z}) \;,
\ee
where $t$ denotes the generator of $H^1(S^1;\mathbbm{Z})$.
$L_T$ is non-degenerate and the two summands in \eqref{EqDecompTorsCohomMS1} are Lagrangian. The restriction of a quadratic refinement to an isotropic subspace is linear, so  $q_c(u \cup t)$ is a linear functional on the first summand of \eqref{EqDecompTorsCohomMS1}. The linking pairing allows us to identify it with an element $\tilde{s}$ in the second summand. It follows that $q_c(\tilde{s} + u \cup t)$ is independent of $u$. 

Let us compute $q_c(s)$. This is especially easy, because $s$ extends to $W_0 = M \times D^2$. The Wu class of $W_0$ vanishes and we can take $\lambda = 0$. This implies that $q_c(s) = 0$ for all $s \in H^{2\ell+2}_{\rm tors}(M;\mathbbm{Z})$. We deduce that $q_c(\tilde{s} + u \cup t) = 0$ for all $u$. Computing $q_c$ using $W_0$ is possible thanks to our careful choice for the integral lift of the Wu class $\lambda$ at the beginning of this section. Indeed, $\lambda$ extends to a lift $\lambda_C$ on $W_C = W \ltxcup \bar{W}_0$, and \eqref{EqDefQROnW} vanishes modulo 1 on $W_C$ by an argument similar to the one in the proof of Proposition \ref{PropIndepBoundMan}.

Finally, we can compute the Arf invariant of $q_c$:
\begin{align}
{\rm Gauss}(q_c) = \, & \sum_{s,u} \exp 2\pi i q_c(s + u \cup t) \notag \\
= \, & \sum_{u} \exp 2\pi i q_c(u \cup t) \sum_s \exp 2\pi i L_M(s,u \cup t) = 1 \;,
\end{align}
so ${\rm Arf}(q_c) = 0$. On the first line, we sum on $s \in H^{2\ell+2}_{\rm tors}(M;\mathbbm{Z})$ and $u \in H^{2\ell+1}_{\rm tors}(M;\mathbbm{Z})$. We use the defining property of the quadratic refinement and the fact that $q_c(s) = 0$ to the second line. Then we notice that $\sum_s \exp 2\pi i L_M(s,u \cup t)$ vanishes unless $u = 0$, in which case it is equal to $1$.

We summarize the results of this section into the following proposition.
\begin{proposition}
For $c$ a trivial loop associated to a micro gauge transformation with torsion class $u \in H^{2\ell+1}_{\rm tors}(M;\mathbbm{Z})$, in the presence of a background torsion flux $s \in H^{2\ell+2}_{\rm tors}(M;\mathbbm{Z})$, we have
\be
h(c) = \exp 2\pi i q_c(s + u \cup t) \;,
\ee
for $q_c$ a quadratic refinement of the linking pairing of $M \times S^1$. There exists a unique $\tilde{s} \in H^{2\ell+2}_{\rm tors}(M;\mathbbm{Z})$ such that $h(c) = 1$ for all $u \in H^{2\ell+1}_{\rm tors}(M;\mathbbm{Z})$. 
\end{proposition}

This means the anomaly is absent if we pick the background gauge field to have torsion characteristic $\tilde{s}$ on $M$. For any other value of $s$, there is a $u$ such that $h(c) = -1$ for $c$ a trivial loop, and \eqref{EqHolForm} is ill-defined. As we will see in the next section, \eqref{EqHolForm} determines the theta characteristic $\tilde{\eta}$ required in Section \ref{SecAnomBundle} for the construction of the anomaly line bundle. So if $s \neq \tilde{s}$, we do not know how to define the anomaly line bundle. It seems therefore reasonable to consider $s = \tilde{s}$ as a consistency condition to be imposed in the background. We will always assume that this is the case in the following, and therefore that \eqref{EqHolForm} is well-defined.

\subsection{The physical theta characteristic}

\label{SecPhysCharSD}

In this section, we will define the theta characteristic $\tilde{\eta}$ appearing in the definition of the anomaly line bundle in Section \ref{SecAnomBundle}. Recall that holomorphic line bundles over a complex torus with a given first Chern class are classified by their \emph{theta characteristic} $\eta$. $\eta$ can be seen as a point in the (dual) torus, i.e. a vector in the universal covering of the torus modulo an integral vector. In our case, the fibers of $\mathcal{F}$ are of the form $\Omega^{2\ell+1}_{\rm closed}(M)/\Omega^{2\ell+1}_\mathbbm{Z}(M)$ and are trivially fibered over the complex torus $\mathcal{H}^{2\ell+1}(M)/\mathcal{H}^{2\ell+1}_\mathbbm{Z}(M)$, the quotient of the space of harmonic $2\ell+1$-forms on $M$ by the lattice of harmonic forms with integral periods. We write $\Lambda := \mathcal{H}^{2\ell+1}_\mathbbm{Z}(M)$ in the following. The complex structure is provided by the Hodge star operator, and the curvature of the line bundle over the torus is $R_\Omega$ as defined in Section \ref{SecConnAnBun}.

The theta characteristic encodes the holonomies of the holomorphic connection on the line bundle along a basis of straight loops in the torus. In our case, these straight loops are a special class of vertical loops, lifting to paths in $\mathcal{B}$ of the form 
\be
\label{EqLiftVertStraightLoop}
p_A = (tA, g) \;, \quad t \in I \;,
\ee
where $g$ is a fixed metric and $A \in \Lambda$. Consider now three elements of $\Lambda$ satisfying $A_1 + A_2 = A_3$, the associated paths $p_i \subset \mathcal{B}$ and the loops $c_i \in \mathcal{F}$ they cover. We have
\be
\label{EqRelHolCurvA}
{\rm hol}_\mathscr{A}(c_1) {\rm hol}_\mathscr{A}(c_2) {\rm hol}_\mathscr{A}(c_3)^{-1} = \exp \int_{\Delta_{123}} R_\Omega = \exp -\pi i \int_M A_1 \wedge A_2 \;,
\ee
where $\Delta_{123}$ is the triangle defined by the three vectors. Define the function $q(A) = -\frac{1}{2\pi i} \ln {\rm hol}_\mathscr{A}(c)$, where $A \in \Lambda$ and $c$ is the associated loop in $\mathcal{F}$. Recall that quadratic refinement are $\mathbbm{R}/\mathbbm{Z}$-valued functions on an abelian group satisfying \eqref{EqDefQR}. We will call \emph{refinements} $\mathbbm{R}/\mathbbm{Z}$-valued functions satisfying only the first equation of \eqref{EqDefQR}.
\eqref{EqRelHolCurvA} tells us that $q$ is a refinement of half the wedge product pairing on $\Lambda$:
\be
q(A_1 + A_2) - q(A_1) - q(A_2) = \frac{1}{2} \int_M A_1 \wedge A_2 =: L_{\wedge}(A_1, A_2) \quad {\rm mod} \; 1 \;.
\ee
Refinements of $L_{\wedge}(A_1, A_2)$ can be parameterized by theta characteristics as follows. Fix a Lagrangian decomposition $\Lambda = \Lambda_+ \oplus \Lambda_-$ with respect to the wedge product pairing, and write $A_+$ and $A_-$ for the components of $A$ on $\Lambda_+$ and $\Lambda_-$. Any refinement can be expressed as 
\be
\label{EqDefQRChar}
\int_M \left( \frac{1}{2} A_+ \wedge A_- + \eta \wedge A \right)
\ee
for a certain $\eta \in \mathcal{H}^{2\ell+1}(M)/\Lambda$. 

We will show in Proposition \ref{ProphcQuadrRef} that $h(c)$ satisfies as well the relation \eqref{EqRelHolCurvA}. Therefore, we would like to define $\tilde{\eta}$ to be such that
\be
\label{EqDefPhysChar}
\exp 2 \pi i \int_M \left( \frac{1}{2} A_+ \wedge A_- + \tilde{\eta} \wedge A \right) = h(c)^{-1} \;.
\ee
Before doing so, there is an important point to check. The construction of Section \ref{SecAnomBundle} is possible only if $\tilde{\eta}$ is half-integral, i.e. $\tilde{\eta} \in \frac{1}{2} \Lambda/\Lambda$. Indeed, only for such $\tilde{\eta}$ can we interpret the theta function $\theta^{\tilde{\eta}}$ as the pull-back to $\mathcal{B}_0$ of the section of a holomorphic line bundle over $\mathcal{A}$. This half-integrality property will be shown in Proposition \ref{ProphStrVertLoopPm1}, and it is therefore consistent to define $\tilde{\eta}$ by \eqref{EqDefPhysChar}. We will also show that $q$ satisfies the second equation of \eqref{EqDefQR}, so it is in fact a quadratic refinement.

Of course, the definition of $\tilde{\eta}$ above immediately implies that \eqref{EqHolFormCl} holds for straight vertical loops (see for instance Section 3.1 of \cite{MR2062673}). We will show that it holds for all vertical loops in Section \ref{SecVertLoops}.

\section{Proof of the anomaly formula}

\label{SecProofAnForm}

In \cite{Monnier2011a} we gave strong evidence that equation \eqref{EqHolFormCl}, identifying $h$ with the holonomy function of the anomaly line bundle, holds for background field configurations with vanishing (form) gauge field, corresponding to loops contained in the $A = 0$ section of $\mathcal{F}$. The aim of this section is to show that under the assumption that \eqref{EqHolFormCl} indeed holds for this class of loops, then \eqref{EqHolFormCl} holds for all loops in $\mathcal{F}$.

In Section \ref{SecHorizLoopVanGaugField}, we consider horizontal loops with vanishing gauge field, discuss the arguments of \cite{Monnier2011a} and introduce our assumption. In Section \ref{SecCompCurv}, we show that \eqref{EqHolFormCl} holds for contractible loops, which amounts to checking that the curvatures of $\nabla_{\mathscr{A}}$ and of the Chern-Simons connection defined by \eqref{EqHolForm} agree. This allows us to show that if two loops are homotopic and \eqref{EqHolFormCl} holds for one of them, then it holds for the other. In Section \ref{SecVertLoops}, we turn to vertical loops, whose holonomies describe the gauge global anomaly and show that \eqref{EqHolFormCl} holds for such loops. In Section \ref{SecHorLoopVanWuLift}, we show that \eqref{EqHolFormCl} holds for generic horizontal loops. We conclude that \eqref{EqHolFormCl} holds for all loops in Section \ref{SecHolFormGenLoops}.

\subsection{Horizontal loops with vanishing gauge field and an assumption}

\label{SecHorizLoopVanGaugField}

Consider a loop $c \subset \mathcal{F}$ lifting to a path in $\mathcal{B} = \Omega_{\rm closed}^{2\ell+1}(M) \times \mathcal{M}$ of the form $(0, p_2)$, i.e. the background gauge field $A \in \Omega_{\rm closed}^{2\ell+1}(M)$ vanishes over each point of $c$ as a differential form. The gluing constraints \eqref{EqCondLoopLift} on the path lift require $p_2(1) = \phi^\ast p_2(0)$ for some diffeomorphism $\phi$. In \cite{Monnier2011a}, it was shown that the holonomies along $c$ of $\nabla_{\mathscr{A}^4}$, the connection on the fourth tensor power of $\mathscr{A}$, are given by the formula
\be
\frac{1}{2\pi i} {\rm hol}_{\mathscr{A}^4}(c) = \lim_{\epsilon \rightarrow 0}\frac{1}{2} \int_W \left(L(TW) - 4F_W^2 \right) \;,
\ee
where $2F_W$ is a form lift of the Wu class vanishing on $\partial W$. It was then argued that the holonomies of $\nabla_{\mathscr{A}}$ along $c$ could be obtained by simply dividing the right-hand side by $4$:
\be
\label{EqHolLoopVanGaugField}
\frac{1}{2\pi i} {\rm hol}_{\mathscr{A}}(c) \stackrel{?}{=} \lim_{\epsilon \rightarrow 0} \frac{1}{8} \int_W \left(L(TW) - 4F_W^2 \right) \;,
\ee
but this was not completely proven. \eqref{EqHolLoopVanGaugField} is obviously equivalent to \eqref{EqHolFormCl} restricted on horizontal loops with vanishing gauge field.

We showed in \cite{Monnier2011a} that as a topological line bundle, $\mathscr{A}^4$ admits only two fourth roots, $\mathscr{A}$ and $\mathscr{A} \otimes \mathscr{T}$, where $\mathscr{T}$ is an order 2 torsion line bundle already mentioned in Section \ref{SecAnomBundle}. We can therefore characterize the ambiguity in \eqref{EqHolLoopVanGaugField} as follows: either \eqref{EqHolLoopVanGaugField} holds, or \eqref{EqHolLoopVanGaugField} holds only if the right-hand side is multiplied by a $\mathbbm{Z}_2$-valued character of the mapping class group of $M$, describing the holonomies of the flat connection on $\mathscr{T}$. \eqref{EqHolLoopVanGaugField} not being true would signal a global gravitational anomaly of type IIB supergravity. \eqref{EqHolLoopVanGaugField} is also required for the cancellation of the five-brane worldvolume anomalies \cite{Monnierb}. In the rest of the paper, we will therefore make the following assumption:
\begin{assumption}
\label{AssumpHolHorLoops}
\eqref{EqHolLoopVanGaugField} holds, or equivalently \eqref{EqHolFormCl} holds for $c$ a horizontal loop with vanishing gauge field. 
\end{assumption}
Given this assumption, we will show that \eqref{EqHolFormCl} holds for all loops in $\mathcal{F}$.

\subsection{Contractible loops}

\label{SecCompCurv}

In this section, we show that \eqref{EqHolFormCl} holds for a contractible loop $c$. As $c$ is contractible, we can fill it with a 2-dimensional disk $D^2$ in $\mathcal{F}$. We can lift this disk to a disk in $\mathcal{B}$. The lift determines a metric $g_x$ on the fiber $M_x$ of $W = M \times D^2$ above $x \in D^2$. We pick a flat metric $g_{D^2}$ on $D^2$ and construct the metric $g_W = g_x \oplus g_{D^2}/\epsilon$. The construction of Section \ref{SecLoopMapTor} also determines a differential cocycle $\check{A}_W = (f_W, h_W, F_W)$ from the embedding $D^2 \subset \mathcal{B}^{\rm N}$.

\begin{proposition}
\label{PropCheckContrLoop}
\eqref{EqHolFormCl} holds when $c$ is a contractible loop. 
\end{proposition}
\begin{proof}
We can reexpress the holonomy of the connection $\nabla_\mathscr{A}$ around $c$ as the integral of its curvature $R_\mathscr{A}$ over the disk bounded by $c$. We therefore would like to prove:
\be
\label{EqHolFormCurvAnBun}
\lim_{\epsilon \rightarrow 0} \frac{1}{8} \int_W \left(  L(TW) - 4 F_W^2 \right) = \frac{1}{2\pi i} \int_{D^2} R_{\mathscr{A}}\;.
\ee
Recall that in Section \ref{SecConnAnBun}, we showed that $R_\mathscr{A}$ pulls back to $R_\Omega + R_\mathcal{M}$ on $\mathcal{B} = \Omega_{\rm closed}^{2\ell+1}(M) \times \mathcal{M}$, where $R_\Omega$ is the pullback of a two-form on $\Omega_{\rm closed}^{2\ell+1}(M)$ and $R_\mathcal{M}$ is the pull-back of a two-form on $\mathcal{M}$. Let us write $\mathscr{D}$ for the determinant line bundle of the signature operator, seen as a line bundle over $\mathcal{M}$. We showed in \cite{Monnier:2010ww} that $R_\mathcal{M}$ is one-fourth of the curvature $R_{\mathscr{D}^{-1}}$ of the Bismut-Freed connection on $\mathscr{D}^{-1}$. Given the setup above, we have
\be
\frac{1}{2\pi i} R_{\mathscr{D}}(x) = -\lim_{\epsilon \rightarrow 0} \frac{1}{2}\left(\int_{M_x} L(TW) \right)^{(2)} \;,
\ee
where $(.)^{(2)}$ denotes the 2-form component. In consequence, we see that the term
\be
\lim_{\epsilon \rightarrow 0} \frac{1}{8} \int_W L_\epsilon(TW)
\ee
in \eqref{EqHolForm} is exactly the integral of $\frac{1}{2\pi i}R_\mathcal{M}$ over the 2-dimensional disk $D^2$.

What remains to be shown is that the second term on the left-hand side of \eqref{EqHolFormCurvAnBun} is the integral of $\frac{1}{2\pi i}R_\Omega$ on $D^2$. As $c$ is homotopically trivial, the curvature $F_c$ of $\check{A}_c$ on $M_c = M \times S^1$ is exact and we can take $F_W$ to be exact: $F_W = dA_W$. Let us expand $dA_W = \partial_i A_W dx^i + d_V A_W$, where $x^i$, $i = 1,2$ are coordinates on $D^2$ and $d_V$ is the differential on  vertical forms. We find that 
\be
\label{EqIntF2ContrCycle}
-\frac{1}{2} \int_W F^2_W = - \int_{D^2} dx^1 dx^2 \int_M \partial_1 A_W \wedge \partial_2 A_W \;.
\ee
Now $-\int_M \partial_1 A_W \wedge \partial_2 A_W$ is easily checked to be $\frac{1}{2\pi i}R_{\Omega}$ applied to the tangent vectors $\partial_1 A_W$ and $\partial_2 A_W$, so \eqref{EqIntF2ContrCycle} is indeed the integral of $\frac{1}{2\pi i}R_\Omega$ over $D^2$.
\end{proof}

\begin{proposition}
\label{PropCompHolHomLoops}
If $h(c) = {\rm hol}(\nabla_\mathscr{A},c)$ for some loop $c$, then $h(c') = {\rm hol}(\nabla_\mathscr{A},c')$ for any loop $c'$ homotopic to $c$.
\end{proposition}
\begin{proof}
Any homotopy between $c$ and $c'$ defines a cylinder $C \subset \mathcal{F}$. We can see this cylinder as a square $I \times I$, with the identification $(0,s) \simeq (1,s)$, and $c'(t) = (t,1)$, $c(t) = (t,0)$. Let $c''$ be the loop defined by the boundary of the square. It is clear that
\be
\label{EqRelHolHomoLoops}
{\rm hol}(\nabla_\mathscr{A},c'') = {\rm hol}(\nabla_\mathscr{A},c')/{\rm hol}(\nabla_\mathscr{A},c) \;,
\ee
because the holonomies of $\nabla_\mathscr{A}$ along the paths $\{(0,s)\}|_{s \in I}$ and $\{(1,s)\}|_{s \in I}$ cancel. 

As $c$ and $c'$ are homotopic, the associated mapping tori $M_{c}$ and $M_{c'}$ are diffeomorphic, so we have a fiber bundle $M_{cc'} = M_c \times I \rightarrow C$ with fiber $M$. Assume $h(c)$ and $h(c')$ are computed from \eqref{EqHolForm} using manifolds $W$ and $W'$ bounded by $M_c$ and $M_{c'}$. Then we can glue $\bar{W}$ and $W'$ to $M_{cc'}$ to obtain a closed manifold $Z$. As we have shown in Proposition \ref{PropIndepBoundMan} that the right-hand side of \eqref{EqHolForm} is an integer on any closed manifold, we obtain the relation 
\be
\label{EqRelHolFormHomoLoops}
h(c'') = h(c')/h(c) \;.
\ee

As $c''$ is a contractible loop, we know that ${\rm hol}(\nabla_\mathscr{A},c'') = h(c'')$. Combining this fact with the hypothesis $h(c) = {\rm hol}(\nabla_\mathscr{A},c)$, we immediately deduce from \eqref{EqRelHolHomoLoops} and \eqref{EqRelHolFormHomoLoops} that 
\be
{\rm hol}(\nabla_\mathscr{A},c') = h(c') \;.
\ee 
\end{proof}

\subsection{Vertical loops}

\label{SecVertLoops}

Recall that vertical loops in $\mathcal{F}$ are loops along which the equivalence class of the metric is constant. The holonomies of $\nabla_\mathscr{A}$ along vertical loops therefore represent pure gauge anomalies. We start by showing that the formula \eqref{EqHolForm} for $h(c)$ simplifies when $c$ is a vertical loop. We recover in \eqref{EqWitGaugGlobAnForm} a formula of Witten \cite{Witten:1996hc} for the global gauge anomalies, up to the fact that an adiabatic limit has to be taken. The presence of the adiabatic limit is natural, because without it the global gauge anomalies would depend on an unphysical choice of a metric on $S^1$ (see Section \ref{SecLoopMapTor} and Proposition \ref{PropHolIndepMetS1}).  

Consider a vertical loop $c$ in the space of background fields $\mathcal{F}$ and assume the usual construction leading to \eqref{EqHolForm} has been carried out. Assume that there is a natural form lift $\lambda_W$ of the Wu class given by a characteristic form. For instance, if we are working in the category of spin manifolds as in \cite{Witten:1996hc}, $\lambda_W$ can be taken to be half the first Pontryagin form of $W$. 
\begin{proposition}
\label{PropWittForm}
For $c$ a vertical loop, we have
\be
\label{EqWitGaugGlobAnForm}
\frac{1}{2\pi i} \ln h(c) = - \frac{1}{8} \lim_{\epsilon \rightarrow 0} \int_W \left( 4F_W^2 - \lambda_W^2 \right) \;.
\ee
\end{proposition}
\begin{proof}
The difference between \eqref{EqHolForm} and \eqref{EqWitGaugGlobAnForm} reads
\be
\label{EqRohlinInv}
\Delta_W = \lim_{\epsilon \rightarrow 0} \frac{1}{8} \int_W \left(L(TW) - \lambda_W^2 \right) \;,
\ee
so we only need to show that this expression is an integer. We already saw in the proof of Proposition \ref{PropIndepBoundMan} that $\Delta_{W_{\rm closed}}$ is an integer for any closed manifold $W_{\rm closed}$. Therefore, if we can find a manifold $W'$ bounded by $\bar{M_c}$ ($M_c$ with its orientation reversed) such that $\Delta_{W'}$ is an integer, then $\Delta_W = \Delta_{W \cup W'} - \Delta_{W'}$ is an integer and the first part of the proposition is proven.

We can take $W' = M \times D^2$, with a constant metric along the fibers and a metric on the disk which is a direct product in a neighborhood $N \simeq M \times S^1 \times [0,1[$ of the boundary. Then over $N$, both $L(TW')$ and $\lambda_{W'}^2$ are pull-back from the boundary and therefore vanish. They define trivial classes in the relative de Rham cohomology of $W'$ and their integral over $W'$ is therefore zero.
\end{proof}

We now focus on straight vertical loops, which are loops in $\mathcal{F}$ lifting to paths of the form $(tA,g_0) \in \mathcal{B}$, with $A$ an element of the lattice $\Lambda$ of harmonic forms with integral periods. Recall that in \eqref{EqDefPhysChar}, we defined the theta characteristic $\tilde{\eta}$ from $h(c)$. For the consistency of the construction of the anomaly line bundle $\mathscr{A}$ in Section \ref{SecAnomBundle}, we need $\tilde{\eta}$ to be an element of $\frac{1}{2} \Lambda/\Lambda$. We will show this in Proposition \ref{ProphStrVertLoopPm1}, using the following two propositions.

\begin{proposition}
\label{ProphcQuadrRef}
Let $A_1, A_2 \in \Lambda$ be the harmonic forms associated to two straight vertical loops $c_1, c_2 \subset \mathcal{F}$ over $g_0$. Let $c_3$ be the straight vertical loop associated to $A_1 + A_2$, and $\Delta_{123} \subset \mathcal{F}$ the triangle determined by $c_1$, $c_2$ and $c_3$. We have
\be
h(c_1)h(c_2)h(c_3)^{-1} = \exp \int_{\Delta_{123}} R_{\Omega} = \exp -\pi i \int_M A_1 \wedge A_2 \;.
\ee
\end{proposition}
\begin{proof}
The composition $c$ of $c_1$, $c_2$ and $c_3^{-1}$ (i.e. $c_3$ with its orientation reversed) is a contractible loop in $\mathcal{F}$. By Proposition \eqref{PropCheckContrLoop}, 
\be
h(c) = {\rm hol}_{\mathscr{A}}(c) = \exp \int_{\Delta_{123}} R_{\Omega} = \exp -\pi i \int_M A_1 \wedge A_2 \;,
\ee
where the second equality comes from the definition \eqref{EqDefROmega} of $R_\Omega$. This proves the proposition.
\end{proof}
To simplify the notations, write $q(c) = -\frac{1}{2\pi i} \ln h(c) \in \mathbbm{R}/\mathbbm{Z}$.
\begin{proposition}
\label{PropQRSecondCond}
Let $c$ be a straight vertical loop and $nc$ its $n$th cover. Then
\be
q(nc) = n^2q(c) \quad {\rm mod} \; 1 \;.
\ee
\end{proposition}
\begin{proof}
Let $p = (tA, g_0)$ the path lifting $c$. Let $W$ be the manifold with boundary associated to this data, $g_W$ its metric and $\check{A}_W = (f_W, h_W, F_W)$ the associated shifted differential cocycle. $F_W$ has integral periods on the boundary, on which the Wu class vanishes. This implies that we can write $F_W = \tilde{F} + \lambda_W/2$, where $\tilde{F}$ has integral periods and $\lambda_W$ is a form lift of the Wu class that vanishes on the boundary. We take $\lambda_W$ in the same de Rham cohomology class as an integral lift $\hat{\lambda}$ of the the Wu class that vanishes on the boundary. Let us use the formula \eqref{EqWitGaugGlobAnForm} and write
\be
q(c) = \frac{1}{2} \int_W \tilde{F} \wedge \left( \tilde{F} + \lambda_W \right) \,.
\ee
In order to compute $q(nc)$ we can use the path $p_n = (ntA, g_0)$. Repeating the steps above, we find
\be
q(nc) = \frac{1}{2} \int_W n\tilde{F} \wedge \left( n\tilde{F} + \lambda_W \right) \,,
\ee
from which we conclude that
\be
q(nc) - n^2q(c) = -\frac{n^2 - n}{2} \int_W \tilde{F} \wedge \lambda_W \quad {\rm mod} \; 1 \;.
\ee
But we saw in the proof of Proposition \ref{PropWittForm} that $\lambda_W$ defines a class $[\lambda_W]_{\rm dR}$ in relative de Rham cohomology. As $\hat{\lambda}$ defines a class in the relative cohomology with integral coefficients that coincides with $[\lambda_W]_{\rm dR}$ in de Rham cohomology, $\lambda_W$ yields an integer when paired with any form with integral periods on $W$. Therefore $\int_W \tilde{F} \wedge \lambda_W$ is an integer. As $\frac{n^2 - n}{2}$ is an integer as well, the proposition is proven.
\end{proof}

Propositions \ref{ProphcQuadrRef} and \ref{PropQRSecondCond} show that the restriction of $h$ to straight vertical loops is a quadratic refinement of the wedge product pairing modulo 2 on $\Lambda$ (see Section \ref{SecGlobTorsAnom} and \cite{Brumfiel1973,Taylor}). We immediately deduce
\begin{proposition}
\label{ProphStrVertLoopPm1}
If $c$ is a straight vertical loop, $h(c) = \pm 1$. 
\end{proposition}
\begin{proof}
In terms of $q$, Proposition \ref{ProphcQuadrRef} asserts that
\be
q(A_1 + A_2) - q(A_1) - q(A_2) = \frac{1}{2} \int_M A_1 \wedge A_2 \quad {\rm mod} \; 1 \;.
\ee
Using this formula with $A_1 = A_2 = A$ and Proposition \ref{PropQRSecondCond}, we find that
\be
q(2A) - 2q(A) = 2q(A) = 0 \quad {\rm mod} \; 1 \;,
\ee 
which proves the proposition.
\end{proof}

Finally, we examine the validity of \eqref{EqHolFormCl}, when $c$ is a generic vertical loop. 
\begin{proposition}
\label{PropVertLoops}
Let $c$ be a vertical loop. Then \eqref{EqHolFormCl} holds, i.e. $h(c) = {\rm hol}(\nabla_\mathscr{A},c)$.
\end{proposition}
\begin{proof}
By the definition \eqref{EqDefPhysChar} of $\tilde{\eta}$, \eqref{EqHolFormCl} holds for straight vertical loops. But any vertical loop is homotopic to a straight vertical loop, so by Proposition \ref{PropCompHolHomLoops}, \eqref{EqHolFormCl} holds for all vertical loops. 
\end{proof}

\subsection{Generic horizontal loops}

\label{SecHorLoopVanWuLift}

Let us now consider a generic horizontal loop $c$ lifting to a path of the form 
\be
\label{EqHorizLoop2}
p(t) = \big(t \phi^\ast(A) + (1-t)A, p_2(t) \big) \;.
\ee

\begin{proposition}
\label{PropGenHorizLoop}
The holonomy of $\nabla_\mathscr{A}$ along a horizontal loop $c$ of the form \eqref{EqHorizLoop2} is given by \eqref{EqHolForm}.
\end{proposition}
\begin{proof}
$c$ trivially retracts to a loop with vanishing gauge fields. An explicit homotopy is given by 
\be
p(t,s) = \big((1-s) t \phi^\ast(A) + (1-s) (1-t)A, p_2(t) \big) \;.
\ee
It follows from Proposition \ref{PropCompHolHomLoops} and Assumption \ref{AssumpHolHorLoops} that \eqref{EqHolForm} computes the holonomies of $\nabla_\mathscr{A}$ on this class of loops.
\end{proof}

\subsection{Generic loops}

\label{SecHolFormGenLoops}

Finally, we can consider a generic loop $c \subset \mathcal{F}$. Proposition \ref{PropHomDecLoops} shows that $c$ is homotopic to the composition of a horizontal loop $c_h$ with a vertical one $c_v$. Propositions \ref{PropGenHorizLoop} and \ref{PropVertLoops} show that \eqref{EqHolForm} computes the holonomy of $\nabla_{\mathscr{A}}$ along $c_h$ and $c_v$, respectively, hence along their composition. Proposition \ref{PropCompHolHomLoops} then allows us to deduce that \eqref{EqHolForm} also computes the holonomy of $\nabla_{\mathscr{A}}$ along $c$. We therefore obtain the main result of this paper:
\begin{theorem}
\label{PropGenLoop}
If Assumption \ref{AssumpHolHorLoops} holds, then \eqref{EqHolFormCl} holds, i.e. \eqref{EqHolForm} computes the holonomy of $\nabla_{\mathscr{A}}$.
\end{theorem}

\section{Type IIB supergravity}

\label{SecIIB}

As an example, we now apply these results to cohomological type IIB supergravity. By \emph{cohomological}, we mean the version of the theory in which the Ramond-Ramond charges and fluxes are classified by the integral cohomology of the underlying manifold. As was shown in \cite{Minasian:1997mm,Witten:1998cd,Moore:1999gb}, the low energy limit of the type II string theories are actually the \emph{K-theoretical} type II supergravities, in which the topological configurations of the Ramond-Ramond gauge fields are classified by the K-theory of spacetime. In order to check the cancellation of global gravitational anomalies in the K-theoretical version of the theory, one would have to extend our results to self-dual fields valued in twisted differential K-theory (see for instance  \cite{Freed:2000ta, Belov:2006xj}). This is outside the scope of the present paper, and we focus on the cohomological version of the theory here. 

In \cite{Monnier2011a}, we checked the cancellation of global gravitational anomalies in cohomological type IIB supergravity. We now use Theorem \ref{PropGenLoop} to study mixed gauge-gravitational anomalies.

\subsection{Setup}

Type IIB supergravity is a 10-dimensional theory containing among other fields $p$-form Ramond-Ramond gauge fields, for $p = 0,2,4$. We will model them by differential cocycles of degree 1, 3 and 5 written $\check{C}_0$, $\check{C}_2$ and $\check{C}_4$, with field strengths $G_1$, $G_3$ and $G_5$. The 4-form gauge field $\check{C}_4$ is self-dual. We can therefore apply our formalism with $\ell = 2$. 

As the theory contains fermions, all the manifolds we will consider belong to the category of spin manifolds. The Wu class of degree 6 vanishes on spin manifolds. This can be seen from its expression in terms of Stiefel-Whitney classes:
\be
\nu_6 = w_4 w_2 + w_4 w_1^2 + w_3^2 + w_3 w_2 w_1 + w_3 w_1^3 + w_2^2 w_1^2 \;. 
\ee
$w_1$, $w_2$ and $w_3$ all vanish for a spin manifold: $w_1$ vanishes because the manifold is orientable, $w_2$ vanishes because of the existence of a spin structure and $w_3 = Sq^1(w_2) + w_2 w_1$ necessarily vanishes as well. Therefore, all the differential cocycles are unshifted. This justifies the claim made in \cite{Monnier2011a} that the lift of the Wu class vanishes.

The role of the background gauge field $\check{A}$ is played by a composite 5-form gauge field formed from the Ramond-Ramond 2-form gauge field $\check{C}_2$ and the B-field $\check{B}$: $\check{A} = \check{B} \cup \check{C}_2$, where $\cup$ is the cup product on differential cohomology. This can for instance be seen in the Bianchi identity for the field strength $G_5$ of $\check{C}_4$, which reads 
\be
\label{EqDiffG5}
dG_5 = H_3 \wedge G_3 \;,
\ee
where $H_3$ and $G_3$ are respectively the field strengths of $\check{B}$ and $\check{C}_2$. Recall that consistency requires $d\check{A} = 0$. This forbids the presence of sources for $\check{B}$ and $\check{C}_2$, i.e. NS5- and D5-branes. Wilson lines for both fields are allowed. (In a treatment of the Ramond-Ramond gauge fields in twisted differential K-theory, this problem is circumvented by replacing the differential with the twisted differential $d_H = d + H_3 \wedge$ and rewriting \eqref{EqDiffG5} as a component of $d_H G = 0$, where $G = G_1 + G_3 + G_5$. In this setting, $\check{A}$ is simply the source current for the Ramond-Ramond fields \cite{Belov:2006xj}, and satisfies indeed the consistency condition $d_H \check{A} = 0$. We will however not pursue this further here.)

Let $\mathcal{M}$ be the space of Riemannian metrics on $M$. Then the space of background fields is $\mathcal{B} = \Omega^5_{\rm closed}(M) \times \mathcal{M}$. The group $\mathcal{G}$ of local transformations is constructed exactly as in Section \eqref{SecActGrpLocTrans} and the space of background fields modulo equivalence is $\mathcal{F} = \mathcal{B}/\mathcal{G}$.

\subsection{Anomaly cancellation}

A loop $c \subset \mathcal{F}$ determines a metric and a degree 6 differential cocycle on a mapping torus $M_c$ with fiber $X$, as shown in Corollary \ref{CorConstrDiffCocMapTor}. It is always possible to extend this data to a manifold $W$ bounded by $M_c$: the relevant cobordism group was shown to vanish in \cite{Kriz:2004tu}. The holonomy of the anomaly connection along $c$ is therefore given by
\be
\label{EqHolFormIIB}
\frac{1}{2\pi i} \ln {\rm hol}_{\mathscr{A}_{SD}}(c) = \lim_{\epsilon \rightarrow 0} \frac{1}{8} \int_W \left(L(TW) - 4 F_W^2 \right) \;,
\ee
where $F_W$ is the field strength of the extension $\check{A}_W$ of the background gauge field to $W$, and we wrote $\mathscr{A}_{SD}$ for the anomaly line bundle with connection of the self-dual 4-form. We have to multiply ${\rm hol}_{\mathscr{A}_{SD}}(c)$ with the holonomies of the connections on the anomaly line bundles of the chiral fermions of the theory. There are two real dilatini (spin $\frac{1}{2}$) of negative chirality and two real gravitini (spin $\frac{3}{2}$) of positive chirality. The global gravitational anomaly of a single real chiral fermion reads \cite{Witten:1985xe, MR861886}
\be
\label{EqHolFormFerm}
\frac{1}{2\pi i} \ln {\rm hol}_{\mathscr{A}_{f}}(c) = \lim_{\epsilon \rightarrow 0} \frac{1}{2} \left( {\rm index}(D_W) - \int_W I_{D_W} \right) \;,
\ee
where $D_W$ is a Dirac operator in dimension $4\ell+4$ related to the chiral Dirac operator in dimension $4\ell +2$ defining the chiral fermionic theory. Writing $I_{1/2}$ and $I_{3/2}$ for the index densities of the Dirac operators on $W$ associated to the dilatini and the gravitini, respectively, the cancellation of local anomalies in type IIB supergravity \cite{AlvarezGaume:1983ig} shows that $-2I_{1/2} + 2I_{3/2} + L(TW) = 0$. We see therefore that there seems to be a remaining anomaly: 
\be
\label{EqHolFormIIBRem}
\frac{1}{2\pi i} \ln {\rm hol}_{\mathscr{A}_{\rm tot}}(c) = - \lim_{\epsilon \rightarrow 0} \frac{1}{2} \int_W F_W^2 \;.
\ee
If the B-field $\check{B}$ and the Ramond-Ramond 2-form gauge field $\check{C}_2$ can be extended separately to $W$, then we can take $F_W = H_3 \wedge G_3$ and $F_W \wedge F_W$ automatically vanishes, as it contains squares of three-forms\footnote{I thank an anonymous referee for pointing this out.}. The vanishing of the anomaly \eqref{EqHolFormIIBRem} would therefore be implied by the vanishing of the spin cobordism group $\Omega_{10}^{\rm spin}(K(3, \mathbb{Z}) \times K(3, \mathbb{Z}))$, classifying bordism classes of spin 10-manifolds endowed with two integral cohomology classes of degree 3. The computation of this cobordism group seems unfortunately non-trivial and is beyond the scope of this paper. (See \cite{Stong:1985vj} for a related computation of the bordism group of spin 11-dimensional manifolds endowed with a degree 4 integral class, relevant to M-theory.)
 
In any case, if the global anomaly \eqref{EqHolFormIIBRem} happened not to cancel, it would only be a curiosity without much consequence. The low energy limit of the type IIB superstring is the K-theoretical version of type IIB supergravity, and this is the theory that the consistency of the superstring requires to have vanishing anomalies. We hope to generalize in a future work the results above to self-dual fields taking value in twisted differential K-theory in order to check the cancellation of global anomalies in the K-theoretical version of type IIB supergravity.

\subsection*{Acknowledgments}

I would like to thank Greg Moore for correspondence and for sharing \cite{Moore2004}, as well as Pavel Mnev for discussions. I also thank the Galileo Galilei Institute for Theoretical physics for hospitality and the INFN for partial support during the completion of this work. This research was supported in part by SNF Grant No.200020-131813/1.

\appendix

\section{More on the anomaly line bundle}

\label{AppMoreAnBun}

In this appendix, we refine the arguments of \cite{Monnier2011} to identify unambiguously the anomaly line bundle of the self-dual field theory. In \cite{Monnier2011}, we showed that the square of the anomaly line bundle $\mathscr{A}$ of the self-dual field, as a line bundle over $\mathcal{F}$, coincides with $(\mathscr{C}^{\tilde{\eta}})^2$. This implies that one of the two following equalities holds:
\be
\label{EqAnomBunAmb}
\mathscr{A} = \mathscr{C}^{\tilde{\eta}} \quad {\rm or} \quad \mathscr{A} = \mathscr{C}^{\tilde{\eta}} \otimes \mathscr{T} \;.
\ee
In the second equality, $\mathscr{T}$ is an explicit line bundle of order 2. 

In \cite{Monnier2011}, we did not exclude that $\mathscr{A}$ could be equal to $\mathscr{C}^{\tilde{\eta}} \otimes \mathscr{T}$. But there is in fact a simple argument showing why this is not possible. Let us first recall some facts from \cite{Monnier2011}. There we showed that the partition function over $\mathcal{B}$ of a pair of self-dual fields of the same chirality is given by 
\be
\label{EqPartFuncPairSD}
Z_{2SD}(g) = (\theta^{\tilde{\eta}}(\tau,A_{\rm harm}))^2 \cdot ({\rm one-loop}) \;.
\ee
In the equation above, $A_{\rm harm}$ is the harmonic part of the gauge field $A$ and $\tau$ is the period matrix determined from the metric through the Hodge star operator on $H^{2\ell+1}(M;\mathbbm{R})$ (see Section \ref{SecRelModGeom}). The second factor is a one-loop determinant, which we could show is non-vanishing and is the section of a trivial line bundle over $\mathcal{F}$. (We were forced to consider a pair of self-dual fields because there does not seem to be a Lagrangian formulation of the theory of a single self-dual field able to reproduce the correct one-loop determinant, determined using other techniques in \cite{Belov:2006jd,Monnier:2010ww}.)

Clearly, \eqref{EqPartFuncPairSD} implies that $\mathscr{A} = (\mathscr{C}^{\tilde{\eta}})^2$, as claimed above. But, at least when the vanishing of the source $\check{C}$ of $\check{A}$ does not constrain the metric of $M$, so that $\mathcal{M}$ is the space of Riemannian metrics on $M$, we can deduce more. In this case, $\mathcal{B}$ is simply connected (it is in fact contractible), there can be no ambiguity other than a global sign when taking the square root of \eqref{EqPartFuncPairSD} and we necessarily have 
\be
\label{EqPartFuncSD}
Z_{SD}(g) = \pm \theta^{\tilde{\eta}}(\tau,A_{\rm harm}) \cdot ({\rm one-loop})^{1/2} \;.
\ee 
The global sign ambiguity is irrelevant. As by definition, $\theta^{\tilde{\eta}}$ is the pull-back of a section of $\mathscr{C}^{\tilde{\eta}}$, we conclude that
\be
\label{EqTopClassAnBunApp}
\mathscr{A} = \mathscr{C}^{\tilde{\eta}} \;.
\ee
When the vanishing $\check{C}$ imposes non-trivial constraints on the metric, the argument above fails because $\mathcal{M}$ is not necessarily contractible. Studying the topology of $\mathcal{M}$ can be done only on a case by case basis, which is beyond the scope of this paper. It seems anyway unlikely that the extra constraint on the metric would modify \eqref{EqTopClassAnBun}.

\section{Derivation of the Kähler potential of the theta line bundle}

\label{AppDerKahlPot}

In this appendix, we derive the Kähler potential \eqref{EqKahlPotThetBun} from the theta transformation formula, i.e. the functional equation satisfied by Siegel theta functions under the action of ${\rm Sp}(2n,\mathbbm{Z})$ on $\mathcal{B}_0$. The latter reads
\begin{align}
\label{EqThetaTransForm}
\theta^{\eta}(z(C\tau + D)^{-1}, \gamma.\tau) & \, =  \kappa_\eta(\gamma) \det(C\tau + D)^{1/2} \exp (\pi i z_i ((C\tau + D)C)^{ij}z_j) \theta^\eta(z,\tau) \notag \\
& \,=:  \xi^\eta(\gamma,z,\tau) \theta^\eta(z,\tau) \;.
\end{align}
Here, $\gamma$ is an element of $\Gamma_\eta \subset {\rm Sp}(2n,\mathbbm{Z})$, written in block diagonal form as
\be
\gamma = \left( \begin{array}{cc} A & B \\ C & D \end{array} \right) \;.
\ee
$\gamma.\tau$ is the image of $\tau$ under the standard action of ${\rm Sp}(2n,\mathbbm{Z})$ on the Siegel upper half-space $\mathcal{C}$. $\kappa_\eta(\gamma)$ is a complex number of unit modulus, whose precise definition depends on a choice of square root for $\det(C\tau + D)^{1/2}$ (see for instance Chapter 8 of \cite{MR2062673}).

The theta function can be seen as the section $s_{\theta}^\eta$ of a $\Gamma_\eta$-equivariant hermitian line bundle $\mathcal{L}$ over $\mathcal{B}_0$. Recall that an equivariant hermitian bundle over $\mathcal{B}_0$ is a hermitian line bundle over $\mathcal{B}_0$ equipped with a lift of the action of $\mathcal{G}_0$ by unitary bundle automorphisms. Given a trivialization of $\mathcal{L}$, we get a family of trivializations indexed by $\mathcal{G}_0$ by acting with the corresponding bundle automorphisms. \eqref{EqThetaTransForm} can then be interpreted as a change of trivialization between two trivializations related by the element $\gamma$. The factor $\xi^\eta(\gamma,z,\tau)$ is the fiberwise endomorphism relating the two trivializations. 

As $|\xi^\eta(\gamma,z,\tau)| \neq 1$, we see that the trivialization of $\mathcal{L}$ in which the theta section takes its usual holomorphic form is not unitary. We would like here to find an $\mathbbm{R}_+$-valued function $f$ over $\mathcal{B}_0$ describing a change of trivialization to a unitary trivialization of $\mathcal{L}$. In the unitary trivialization, the theta section takes the form $s^\eta_\theta = f \theta^\eta$. Its norm is 
\be
|s^{\eta}_\theta|^2_{\mathcal{L}} = f^2 |\theta^\eta|^2 \;,
\ee 
so comparing with \eqref{EqDefKahlPot}, we see that $K = -2\ln f$. Moreover, as the new trivialization is unitary and the bundle automorphisms defining the equivariant structure are unitary as well, $f^2 |\theta^\eta|^2$ should be invariant under the action of $\Gamma_\eta$. This implies
\be
|\xi^\eta(\gamma,z,\tau)| = \frac{f\left(z(C\tau + D)^{-1}, \gamma.\tau\right)}{f(z,\tau)} =: \delta_{\gamma} f \;,
\ee
i.e. $f$ trivializes $|\xi^\eta|$, seen as a group 1-cocycle. Recall that $h = i(\tau-\bar{\tau})^{-1}$. It is easy to check that
\be
\label{EqVarDetH}
\delta_\gamma (\det h) = |\det(C\tau + D)|^2 \;.
\ee
With a bit more work (see for instance the proof of Proposition 8.6.1 in \cite{MR2062673}), one can show that
\be
\label{EqVarzhz}
\delta_\gamma \exp -2\pi (z_i-\bar{z}_i)h^{ij}(z_j-\bar{z}_j) = |\exp (2\pi i z_i ((C\tau + D)C)^{ij}z_j)|^2 \;.
\ee
We deduce from \eqref{EqVarDetH} and \eqref{EqVarzhz} that the Kähler potential is given by
\be
K = \pi (z_i-\bar{z}_i)h^{ij}(z_j-\bar{z}_j) - \frac{1}{2} \ln \det h \;,
\ee
as was claimed in \eqref{EqKahlPotThetBun}. Now recall that the Dolbeault operator reads
\be
\partial = dz_i \partial_{z_i} + d\tau_{ij} \partial_{\tau_{ij}} 
\ee
and a straightforward computation yields the curvature \eqref{EqCurvThetaBun}.

{
\small
%\bibliographystyle{../../Bibliographie/BibStyle/utphys}
%\bibliography{../../Bibliographie/References}

\providecommand{\href}[2]{#2}\begingroup\raggedright\endgroup

}

\end{document}